\documentstyle[cmll, amsmath, amssymb, relsize, bussproofs, url, appendix]{mscs}

\title[A note on undecidability of propositional non-associative linear logics]
  {A note on undecidability of propositional non-associative linear logics}
\author[Hiromi Tanaka]
  {H\ls I\ls R\ls O\ls M\ls I\ns
   T\ls A\ls N\ls A\ls K\ls A$^1$\\
$^1$ Department of Philosophy, Keio University\addressbreak
      Mita, Minato-ku, Tokyo 108-8345, Japan.\addressbreak
htanaka304@gmail.com}

\date{}
\pubyear{}
\pagerange{\pageref{firstpage}--\pageref{lastpage}}
\volume{}

\newtheorem{Thm}{Theorem}[section]  
\newtheorem{Cor}[Thm]{Corollary}
\newtheorem{Lem}[Thm]{Lemma}
\newtheorem{Def}[Thm]{Definition}
\newtheorem{Remark}[Thm]{Remark}
\newtheorem{Prop}[Thm]{Proposition}
\newtheorem{Exm}[Thm]{Example}
\newcommand{\negr}[1]{\sim\!#1}
\newcommand{\negl}[1]{-#1}
\newcommand{\negrl}[1]{\sim\!\!-#1} 
\newcommand{\neglr}[1]{-\!\!\sim\!#1} 
\newcommand{\N}[1]{\,N\,#1}
\newcommand{\NN}[1]{\,N^*\,#1}
\newcommand{\id}{\mathop{\mathrm{id}}\nolimits}
\newcommand{\exch}{\mathop{\mathrm{e}}\nolimits}
\newcommand{\cont}{\mathop{\mathrm{c}}\nolimits}
\newcommand{\weak}{\mathop{\mathrm{w}}\nolimits}
\newcommand{\asso}{\mathop{\mathrm{a}}\nolimits}
\newcommand{\dbackslash}[1]{\backslash\!\!\backslash#1}
\newcommand{\dslash}[1]{/\!\!/#1}
\newenvironment{bprooftree}
  {\leavevmode\hbox\bgroup}
  {\DisplayProof\egroup}
\renewcommand{\labelenumi}{(\roman{enumi})}
\begin{document}

\label{firstpage}
\maketitle

\begin{abstract}
We introduce a non-associative and non-commutative version of propositional intuitionistic linear logic, called   propositional non-associative non-commutative intuitionistic linear logic ($\mathbf{NACILL}$ for short). 
We prove that $\mathbf{NACILL}$ and any of its extensions by the rules of exchange and/or contraction are undecidable. 
Furthermore, we introduce two types of classical versions of $\mathbf{NACILL}$, i.e., an involutive version of $\mathbf{NACILL}$ and a cyclic and involutive version of $\mathbf{NACILL}$. 
We show that both of these logics are also undecidable. 
\end{abstract}

%\tableofcontents

\section{Introduction}
\label{aaa} 
Lincoln \emph{et al.} (1992) showed that propositional classical linear logic, propositional non-commutative classical linear logic, and their intuitionistic versions are all undecidable, which is known as one of the remarkable results in the early period of linear logic. 
To the best of our knowledge, non-associative versions of propositional linear logic, called \emph{propositional non-associative linear logics}, have not been well-investigated and the decision problems for these logics have not been settled so far. 
In this paper, we introduce propositional non-associative linear logics in natural ways, and show that the decision problems for some of them are undecidable, as a continuation of the above well-known work by Lincoln \emph{et al}.

On the other hand, various research results on non-associative logics have been accumulated in the research fields of substructural logic and algebraic logic, such as studies on \emph{full non-associative Lambek calculus} ($\mathbf{FNL}$) \cite{GO10,GJ13}. 
The main difference between the linear-logical setting and the standard setting of substructural logic is the presence of the linear-logical modal operator, which is often called  ``exponential" or  ``bang". 
This operator plays an important role in terms of computational complexity. 
For instance, as shown in \cite{LMSS92}, the modality-free fragment of propositional linear logic is PSPACE-complete, in contrast to the undecidability of propositional linear logic. 
This means that the rich expressive power of linear logic is caused by the linear-logical modal operator.
 
As the basis for considering various versions of propositional non-associative linear logics in this paper, we introduce \emph{propositional non-associative non-commutative intuitionistic linear logic} ($\mathbf{NACILL}$), by enriching $\mathbf{FNL}$ with a sort of modal operator. 
The modality of $\mathbf{NACILL}$ is employed to limit the use of the rules of weakening, contraction, exchange and associativity. 

We show that the decision problem for $\mathbf{NACILL}$ is undecidable.  
Our proof of the undecidability of $\mathbf{NACILL}$ is based on the ideas in \cite{Ch15} and \cite{LMSS92}. 
Chvalovsk\'{y} (2015) had a breakthrough result, by proving that the finitary consequence relation in $\mathbf{FNL}$ is undecidable. 
In view of this, we show the undecidability of $\mathbf{NACILL}$, by encoding the finitary consequence relation in $\mathbf{FNL}$ into the provability in $\mathbf{NACILL}$. 
The idea of this encoding comes from \cite{LMSS92}.     
To confirm that our encoding is sound and faithful, we show that $\mathbf{NACILL}$ is a strongly conservative extension of $\mathbf{FNL}$. 

\if0
It seems difficult to show that $\mathbf{NACILL}$ is a strongly conservative extension of $\mathbf{FNL}$, using only proof-theoretic methods; hence, we show this with the aid of algebraic techniques. 
In fact, Chvalovsk\'{y} (2015)  also proved that the finitary consequence relation in each of the extensions of $\mathbf{FNL}$ by the rules of contraction and exchange is undecidable. 
Furthermore, the techniques we use in the proof of the undecidability of $\mathbf{NACILL}$ can be easily applied to the case where contraction and exchange are added to $\mathbf{NACILL}$. In view of these facts,
\fi
In fact, Chvalovsk\'{y} (2015)  also proved that the finitary consequence relation in each of the extensions of $\mathbf{FNL}$ by the rules of contraction and exchange is undecidable. 
Furthermore, the techniques we use in the proof of the undecidability of $\mathbf{NACILL}$ can be easily applied to the case where contraction and exchange are added to $\mathbf{NACILL}$. 
In view of these facts, we prove the following generalized form of the undecidability of $\mathbf{NACILL}$, which is the main result of this paper: 

``Every extension of $\mathbf{NACILL}$ by a (possibly empty) subset of the rules of contraction and exchange is undecidable (Theorem~\ref{main1} in Section~\ref{mainsection})."

Here, in particular, we stress the undecidability of $\mathbf{NACILL}$ with exchange and contraction.
This is in sharp contrast to the decidability of intuitionistic linear logic with contraction which was proved in \cite{OT99,LS15}.

In addition, we introduce two types of non-associative and non-commutative versions of propositional classical linear logic: one is an involutive version of $\mathbf{NACILL}$ (denoted by $\mathbf{NACCLL}^-$) and the other is a cyclic and involutive version of $\mathbf{NACILL}$ (denoted by $\mathbf{NACCLL}$). 
We show that both of these logics are also undecidable.
The proofs of the undecidability of $\mathbf{NACCLL}^-$ and $\mathbf{NACCLL}$ is similar to the proof of the main theorem. 
Hence, we do not describe the proofs of the undecidability of $\mathbf{NACCLL}^-$ and $\mathbf{NACCLL}$ in detail.

At the end of the introduction, we summarize the contents of this paper. 
Section \ref{pre} is divided into two parts. 
In the first half of the section, we recall the syntax of $\mathbf{FNL}$ and introduce the syntax of $\mathbf{NACILL}$. 
In the second half of the section, we describe the algebraic semantics for $\mathbf{FNL}$ and $\mathbf{NACILL}$, and recall some algebraic notions, such as nuclei, residuated frames, and Dedekind-MacNeille completions; these are useful for showing the main theorem in Section~\ref{mainsection}. 
Section~\ref{classical} is devoted to a short discussion of the undecidability of $\mathbf{NACCLL}^-$ and $\mathbf{NACCLL}$.  
In Section~\ref{future}, we conclude this paper with some open questions.  

In \ref{gentzen}, independently of the main theorem, we discuss cut elimination for propositional non-associative intuitionistic linear logics from an algebraic standpoint; concretely, we prove that $\mathbf{NACILL}$ and all its extensions by the rules of weakening, exchange and contraction admit cut elimination, using a sort of residuated frame. 
  
\section{Preliminaries}
\label{pre}
We start with the syntax of \emph{full non-associative Lambek calculus} $\mathbf{FNL}$.
Our explanation is based on \cite{GO10,GJ13,Ch15,GJ17}. 

The \emph{language} $\mathcal{L}$ of $\mathbf{FNL}$ consists of the binary connectives $\land, \lor, \cdot, \backslash,/$ and the constant $1$. 
We fix a countable set of propositional variables. 
We denote it by $\mathsf{Var}$. 
An \emph{$\mathcal{L}$-formula} is a term in the language $\mathcal{L}$ over $\mathsf{Var}$. 
The set of $\mathcal{L}$-formulas is denoted by $Fm_{\mathcal{L}}$. 
$Fm^{\circ}_{\mathcal{L}}=(Fm^{\circ}_{\mathcal{L}},\circ,\varepsilon)$ denotes the free unital groupoid generated by the set $Fm_{\mathcal{L}}$. 
An {\emph{$\mathcal{L}$-structure} is simply an element of $Fm^{\circ}_{\mathcal{L}}$. 
We denote by $S_{Fm^{\circ}_{\mathcal{L}}}$ the set of unary linear polynomials over $Fm^{\circ}_{\mathcal{L}}$. 
An \emph{$\mathcal{L}$-sequent} is an element of $Fm_{\mathcal{L}}^{\circ} \times Fm_{\mathcal{L}}$. 
For readability, given an $\mathcal{L}$-sequent $(x,a)$, we always denote it by $x \Rightarrow a$. 
The sequent calculus for $\mathbf{FNL}$ consists of the initial sequents and the inference rules given in Figure \ref{inf1}. 
Letters $a,b,c$ range over $\mathcal{L}$-formulas, $x,y$ over $\mathcal{L}$-structures, and $u$ over unary linear polynomials over $Fm_{\mathcal{L}}^{\circ}$ in Figure~\ref{inf1}. 
Likewise, $u[x]$ stands for the image of $x$ under $u$. 
Given an $\mathcal{L}$-sequent $s$, a \emph{proof} of $s$ and the \emph{provability} of $s$ in $\mathbf{FNL}$ are defined as usual. 
Specifically, given an $\mathcal{L}$-formula $a$, we say that $a$ is provable in $\mathbf{FNL}$ if the sequent $\varepsilon \Rightarrow a$ is provable in $\mathbf{FNL}$. 

Next, we review the consequence relation in $\mathbf{FNL}$, for which we use the notation $\vdash_{\mathbf{FNL}}$. 
Let $\Phi \cup \{x \Rightarrow a\}$ be a set of $\mathcal{L}$-sequents. 
We say that $x \Rightarrow a$ \emph{is provable in ${\mathbf{FNL}}$ from $\Phi$} and write $\Phi \vdash_{\mathbf{FNL}} x \Rightarrow a$ if $x \Rightarrow a$ is provable in the sequent calculus obtained from ${\mathbf{FNL}}$ by adding $s$ as an initial sequent for each $s \in \Phi$. 

\begin{figure}[t]
\begin{description}
\item[Initial sequents:]
\[
\begin{bprooftree}
\AxiomC{}
\RightLabel{(Id)}
\UnaryInfC{$a \Rightarrow a$}
\end{bprooftree}
\begin{bprooftree}
\AxiomC{}
\RightLabel{($\Rightarrow 1$)}
\UnaryInfC{$\varepsilon \Rightarrow 1$}
\end{bprooftree}
\]
\item[Cut:]
\[
\begin{bprooftree}
\AxiomC{$x \Rightarrow a$}
\AxiomC{$u[a] \Rightarrow c$}
\RightLabel{(cut)}
\BinaryInfC{$u[x] \Rightarrow c$}
\end{bprooftree}
\]
\item[Rules for logical connectives:]
\[
\begin{bprooftree}
\AxiomC{$u[\varepsilon] \Rightarrow c$}
\RightLabel{$(1\Rightarrow)$}
\UnaryInfC{$u[1] \Rightarrow c$}
\end{bprooftree}
\begin{bprooftree}
\AxiomC{$x \Rightarrow a$}
\AxiomC{$u[b]\Rightarrow c$}
\RightLabel{$(\backslash \Rightarrow)$}
\BinaryInfC{$u[x \circ (a \backslash b)] \Rightarrow c$}
\end{bprooftree}
\]
\[
\begin{bprooftree}
\AxiomC{$a \circ x \Rightarrow b$}
\RightLabel{$(\Rightarrow \backslash)$}
\UnaryInfC{$x \Rightarrow a \backslash b$}
\end{bprooftree}
\begin{bprooftree}
\AxiomC{$u[a \circ b] \Rightarrow c$}
\RightLabel{$(\cdot \Rightarrow)$}
\UnaryInfC{$u[a \cdot b] \Rightarrow c$}
\end{bprooftree}
\begin{bprooftree}
\AxiomC{$x \Rightarrow a$}
\AxiomC{$y \Rightarrow b$}
\RightLabel{$(\Rightarrow \cdot)$}
\BinaryInfC{$x \circ y \Rightarrow a \cdot b$}
\end{bprooftree}
\]
\[
\begin{bprooftree}
\AxiomC{$x \circ a \Rightarrow b$}
\RightLabel{$(\Rightarrow /)$}
\UnaryInfC{$x \Rightarrow b/a$}
\end{bprooftree}
\begin{bprooftree}
\AxiomC{$x \Rightarrow a$}
\AxiomC{$u[b] \Rightarrow c$}
\RightLabel{$(/ \Rightarrow)$}
\BinaryInfC{$u[(b/a) \circ x] \Rightarrow c$}
\end{bprooftree}
\]
\[
\begin{bprooftree}
\AxiomC{$u[a] \Rightarrow c$}
\RightLabel{$(\land \Rightarrow)$}
\UnaryInfC{$u[a \land b] \Rightarrow c$}
\end{bprooftree}
\begin{bprooftree}
\AxiomC{$u[b] \Rightarrow c$}
\RightLabel{$(\land \Rightarrow)$}
\UnaryInfC{$u[a \land b] \Rightarrow c$}
\end{bprooftree}
\begin{bprooftree}
\AxiomC{$x \Rightarrow a$}
\AxiomC{$x \Rightarrow b$}
\RightLabel{$(\Rightarrow \land)$}
\BinaryInfC{$x \Rightarrow a \land b$}
\end{bprooftree}
\]
%\vspace{6pt}
\[
\begin{bprooftree}
\AxiomC{$u[a] \Rightarrow c$}
\AxiomC{$u[b] \Rightarrow c$}
\RightLabel{$(\lor \Rightarrow)$}
\BinaryInfC{$u[a \lor b] \Rightarrow c$}
\end{bprooftree}
\begin{bprooftree}
\AxiomC{$x \Rightarrow a$}
\RightLabel{$(\Rightarrow \lor)$}
\UnaryInfC{$x \Rightarrow a \lor b$}
\end{bprooftree}
\begin{bprooftree}
\AxiomC{$x \Rightarrow b$}
\RightLabel{$(\Rightarrow \lor)$}
\UnaryInfC{$x \Rightarrow a \lor b$}
\end{bprooftree}
\]
\end{description}
\caption{Inference rules of $\mathbf{FNL}$}
\label{inf1}
\end{figure}

Moreover, we introduce extensions of $\mathbf{FNL}$ by new inference rules, using terminology from \cite{HT11}. 
Let $R$ be a set of inference rules closed under substitutions.  
The \emph{extension of $\mathbf{FNL}$ by $R$}, which is denoted by ${\mathbf{FNL}}_R$, is the sequent calculus obtained from the sequent calculus for $\mathbf{FNL}$ by adding all the inference rules in $R$. 
The consequence relation in ${\mathbf{FNL}}_R$ is defined in the same way as that in $\mathbf{FNL}$. 
In this paper, we often consider the extensions of $\mathbf{FNL}$ by some of the following structural rules:
\[
\begin{bprooftree}
\AxiomC{$u[x \circ y] \Rightarrow c$}
\RightLabel{$(\exch)$}
\UnaryInfC{$u[y \circ x] \Rightarrow c$}
\end{bprooftree}
\begin{bprooftree}
\AxiomC{$u[\varepsilon] \Rightarrow c$}
\RightLabel{$(\weak)$}
\UnaryInfC{$u[x] \Rightarrow c$}
\end{bprooftree}
\begin{bprooftree}
\AxiomC{$u[x \circ x] \Rightarrow c$}
\RightLabel{$(\cont)$}
\UnaryInfC{$u[x] \Rightarrow c$}
\end{bprooftree}
\begin{bprooftree}
\AxiomC{$u[(x \circ y) \circ z] \Rightarrow c$}
\RightLabel{$(\asso)$}
\doubleLine
\UnaryInfC{$u[x \circ (y \circ z)] \Rightarrow c$}
\end{bprooftree}
\]
The double horizontal line of the rule of (a) means that the sequent under the double horizontal line implies the sequent over the double horizontal line, in addition to the usual meaning. 
The extension of $\mathbf{FNL}$ by the rule of $(\asso)$ (i.e., $\mathbf{FNL}_{\asso}$) is equivalent to the positive fragment of \emph{full Lambek calculus} $\mathbf{FL}$. 

Next, we introduce the syntax of \emph{propositional non-associative non-commutative intuitionistic linear logic} $\mathbf{NACILL}$.
The language $\mathcal{L}_{\oc}$ of $\mathbf{NACILL}$ is obtained from $\mathcal{L}$ by adding the unary connective $\oc$.
Formulas, structures and sequents in the language $\mathcal{L}_{\oc}$ are defined in the same way as those in the language $\mathcal{L}$. 
The set of $\mathcal{L}_{\oc}$-formulas (resp. $\mathcal{L}_{\oc}$-structures) is written by $Fm_{\mathcal{L}_{\oc}}$ (resp. $Fm^{\circ}_{\mathcal{L}_{\oc}}$). 
$S_{Fm^{\circ}_{\mathcal{L}_{\oc}}}$ denotes the set of unary linear polynomials over $Fm^{\circ}_{\mathcal{L}_{\oc}}$. 
The free unital groupoid generated by the set $\{\oc a \mid a \in Fm_{\mathcal{L}_{\oc}}\}$ is denoted by $K_{\mathcal{L}_{\oc}}$. 
A sequent calculus for $\mathbf{NACILL}$ is obtained from the sequent calculus for $\mathbf{FNL}$ by adding all the inference rules in Figure \ref{inf2}. 
We always assume that $a,b,c$ range over $Fm_{\mathcal{L}_\oc}$, $x,y$ over $Fm_{\mathcal{L}_{\oc}}^{\circ}$, $u$ over $S_{Fm_{\mathcal{L}_{\oc}}^{\circ}}$, and $x^{\oc},z^{\oc}$ over $K_{\mathcal{L}_{\oc}}$ in Figures \ref{inf1} and \ref{inf2}, when considering sequent calculi for logics in the language $\mathcal{L}_{\oc}$. 
\begin{figure}[t]
\[
\begin{bprooftree}
\AxiomC{$u[a] \Rightarrow c$}
\RightLabel{$(\oc \Rightarrow)$}
\UnaryInfC{$u[\oc a] \Rightarrow c$}
\end{bprooftree}
\begin{bprooftree}
\AxiomC{$x^{\oc} \Rightarrow a$}
\RightLabel{$(\Rightarrow \oc)$}
\UnaryInfC{$x^{\oc} \Rightarrow \oc a$}
\end{bprooftree}
\begin{bprooftree}
\AxiomC{$u[\varepsilon] \Rightarrow c$}
\RightLabel{$(\oc \weak)$}
\UnaryInfC{$u[x^{\oc}] \Rightarrow c$}
\end{bprooftree}
\]
\[
\begin{bprooftree}
\AxiomC{$u[x^{\oc} \circ x^{\oc}] \Rightarrow c$}
\RightLabel{$(\oc \cont)$}
\UnaryInfC{$u[x^{\oc}] \Rightarrow c$}
\end{bprooftree}
\begin{bprooftree}
\AxiomC{$u[x^{\oc} \circ y] \Rightarrow c$}
\RightLabel{$(\oc \exch)$}
\doubleLine
\UnaryInfC{$u[y \circ x^{\oc}] \Rightarrow c$}
\end{bprooftree}
\]
\[
\begin{bprooftree}
\AxiomC{$u[(x^{\oc} \circ y) \circ z] \Rightarrow c$}
\RightLabel{$(\oc\asso)$}
\doubleLine
\UnaryInfC{$u[x^{\oc} \circ (y \circ z)] \Rightarrow c$}
\end{bprooftree}
\begin{bprooftree}
\AxiomC{$u[(x \circ y) \circ z^{\oc}] \Rightarrow c$}
\RightLabel{$(\oc \asso^*)$}
\doubleLine
\UnaryInfC{$u[x \circ (y \circ z^{\oc})] \Rightarrow c$}
\end{bprooftree}
\]
\caption{Rules for modality}
\label{inf2}
\end{figure}
${\mathbf{NACILL}}_R$ denotes the sequent calculus obtained from the sequent calculus for $\mathbf{NACILL}$ by a set $R$ of inference rules closed under substitutions. 
The consequence relations $\vdash_{\mathbf{NACILL}}$ and $\vdash_{{\mathbf{NACILL}}_{R}}$ are defined in a natural way.
Clearly, the following rule is admissible in $\mathbf{NACILL}$:

\begin{prooftree}
\AxiomC{$u[(x \circ y^{\oc}) \circ z] \Rightarrow c$}
\RightLabel{$(\oc \asso^{**})$}
\doubleLine
\UnaryInfC{$u[x \circ (y^{\oc} \circ z)] \Rightarrow c$}
\end{prooftree}

The following proposition summarizes basic properties of $\mathbf{NACILL}$. 
\begin{Prop}
\label{base}
The following formulas are provable in $\mathbf{NACILL}$:
\begin{enumerate}[(viii)]
\item $\oc 1$,
\item $\oc (a \backslash b) \backslash (\oc a \backslash \oc b)$, 
\item $\oc a \backslash a$,
\item $\oc a \backslash \oc \oc a$,
\item $\oc a \backslash 1$,
\item $\oc a \backslash (\oc a \cdot \oc a)$,
\item $(\oc a \cdot \oc b) \backslash \oc (a \wedge b)$ and $\oc(a \wedge b) \backslash (\oc a \cdot \oc b)$,
\item $(\oc a \cdot b) \backslash (b \cdot \oc a)$ and $(b \cdot \oc a) \backslash (\oc a \cdot b)$,
\item $((\oc a \cdot b) \cdot c) \backslash (\oc a \cdot (b \cdot c))$ and $(\oc a \cdot (b \cdot c)) \backslash ((\oc a \cdot b) \cdot c)$,
\item $((a \cdot b) \cdot \oc c) \backslash (a \cdot (b \cdot \oc c))$ and $(a \cdot (b \cdot \oc c)) \backslash ((a \cdot b) \cdot \oc c)$,
\item $((a \cdot \oc b) \cdot c) \backslash (a \cdot (\oc b \cdot c))$ and $(a \cdot (\oc b \cdot c)) \backslash ((a \cdot \oc b) \cdot c)$.
\end{enumerate}
\end{Prop}
\begin{proof}
All are straightforward to show. 
\if0
For example,  there is the following proof of $((\alpha \cdot \oc \beta) \cdot \gamma) \backslash (\alpha \cdot (\oc \beta \cdot \gamma))$:
\begin{prooftree}
\AxiomC{$a \Rightarrow a$}
\AxiomC{$\oc b \Rightarrow \oc b$}
\AxiomC{$c \Rightarrow c$}
\BinaryInfC{$(\oc b,c) \Rightarrow \oc b \cdot c$}
\BinaryInfC{$(a,(\oc b,c)) \Rightarrow a \cdot (\oc b \cdot c)$}
\UnaryInfC{$(\alpha,(\gamma,\oc \beta)) \Rightarrow \alpha \cdot (\oc \beta \cdot \gamma)$}
\UnaryInfC{$((\alpha,\gamma),\oc \beta) \Rightarrow \alpha \cdot (\oc \beta \cdot \gamma)$}
\UnaryInfC{$(\oc \beta,(\alpha,\gamma)) \Rightarrow \alpha \cdot (\oc \beta \cdot \gamma)$}
\UnaryInfC{$((\oc \beta,\alpha),\gamma) \Rightarrow \alpha \cdot (\oc \beta \cdot \gamma)$}
\UnaryInfC{$((\alpha,\oc \beta),\gamma) \Rightarrow \alpha \cdot (\oc \beta \cdot \gamma)$}
\UnaryInfC{$(\alpha \cdot \oc \beta,\gamma) \Rightarrow \alpha \cdot (\oc \beta \cdot \gamma)$}
\UnaryInfC{$(\alpha \cdot \oc \beta) \cdot \gamma \Rightarrow \alpha \cdot (\oc \beta \cdot \gamma)$}
\UnaryInfC{$\Rightarrow ((\alpha \cdot \oc \beta) \cdot \gamma) \backslash (\alpha \cdot (\oc \beta \cdot \gamma))$}
\end{prooftree}
\fi
\end{proof}

In what follows, we describe the algebraic models for $\mathbf{FNL}$ and $\mathbf{NACILL}$. 
An \emph{$\mathcal{L}$-algebra} (resp. \emph{$\mathcal{L}_{\oc}$-algebra}) is an algebra in the language $\mathcal{L}$ (resp. $\mathcal{L}_{\oc}$), i.e., an algebra of the form $(A,\wedge,\vee,\cdot,\backslash,/,1)$ (resp. $(A,\wedge,\vee,\cdot,\backslash,/,\oc,1)$). 
Given an $\mathcal{L}$-algebra (or $\mathcal{L}_{\oc}$-algebra) $\mathbf{A}$, a map $f \colon {\mathsf{Var}} \rightarrow A$ is called a \emph{valuation} into $\mathbf{A}$. 
This map is uniquely extended to the homomorphism $f$ from $\mathbf{Fm}_{\mathcal{L}}$ (resp. $\mathbf{Fm}_{\mathcal{L}_{\oc}}$) to $\mathbf{A}$, where $\mathbf{Fm}_{\mathcal{L}}$ (resp. $\mathbf{Fm}_{\mathcal{L}_{\oc}}$) denotes the absolutely free algebra in the language $\mathcal{L}$ (resp. $\mathcal{L}_{\oc}$) over $\mathsf{Var}$. 
We also call this homomorphism a valuation. 
Given a structure $x$, $\rho(x)$ stands for the formula obtained from $x$ by replacing $\circ$ by $\cdot$.
In particular, we set $\rho(x)=1$ if $x=\varepsilon$. 

Let $\Phi \cup \{x \Rightarrow a\}$ be a set of sequents. 
Given an algebra $\mathbf{A}$ and a valuation $f$ into $\mathbf{A}$, we write $\Phi \models_{{\mathbf{A}},f} x \Rightarrow a$ if $f(\rho(x)) \leq f(a)$ whenever $f(\rho(y)) \leq f(b)$ for all $y \Rightarrow b \in \Phi$. 
We write $\Phi \models_{\mathbf{A}} x \Rightarrow a$ if $\Phi \models_{{\mathbf{A}},f} x \Rightarrow a$ holds for all valuation $f $ into $\mathbf{A}$. 
Moreover, given a class $\mathcal{K}$ of algebras, we write $\Phi \models_{\mathcal{K}} x \Rightarrow a$ if $\Phi \models_{\mathbf{A}} x \Rightarrow a$ for any $\mathbf{A} \in \mathcal{K}$. 

Next, we briefly recall $r \ell u$-groupoids. 
For more on $r \ell u$-groupoids and related notions, we refer the reader to \cite{GJKO07,GO10,GJ13}.  
\begin{Def}
A \emph{residuated lattice-ordered unital groupoid} (\emph{$r\ell u$-groupoid} for short) is an $\mathcal{L}$-algebra ${\mathbf{G}}=(G,\wedge,\vee,\cdot,\backslash,/,1)$ such that:
\begin{itemize}
\item $(G,\wedge,\vee)$ is a lattice,
\item $(G,\cdot,1)$ is a unital groupoid, and
\item for any $x,y,z \in G$, $xy \leq z$ iff $y \leq x\backslash z$ iff $x \leq z/y$.
\end{itemize} 
\end{Def}

We usually write $xy$ instead of $x \cdot y$. 
The inequation $x \leq y$ holds if and only if the equation $x \wedge y=x$ holds in any lattice. 
In view of this, inequations are always referred to as equations when considering lattice-ordered algebras. 
The class $\mathsf{RLUG}$ of $r\ell u$-groupoids forms a variety (cf. \cite{GO10}).  

For the rules of exchange $(\exch)$, contraction $(\cont)$ and weakening $(\weak)$, we consider the following three identities:
\begin{align*}
xy &\leq yx & x &\leq xx  & x &\leq 1
\end{align*} 
These identities are abbreviated by $\mathsf{e}$, $\mathsf{c}$ and $\mathsf{w}$, respectively.  
Given $R \subseteq \{\exch,\cont,\weak\}$, $\mathsf{RLUG}_{\mathsf{R}}$ denotes the subvariety of $\mathsf{RLUG}$ axiomatized by the set $\mathsf{R}$, where $\mathsf{R}$ denotes the subset of $\{\mathsf{e},\mathsf{c},\mathsf{w}\}$ corresponding to $R$.
For instance, if $R=\{\exch\}$, then $\mathsf{RLUG}_{\mathsf{R}}$ ($=\mathsf{RLUG}_{\mathsf{e}}$) is the variety of commutative $r \ell u$-groupoids.  
In a standard way one proves the following (strong) completeness theorem:
\begin{Lem}[Galatos and Ono (2010)]
\label{algebra1}
Let $R$ be a subset of $\{\exch,\cont,\weak\}$ and $\Phi \cup \{x \Rightarrow a\}$ a set of $\mathcal{L}$-sequents. $\Phi \vdash_{{\mathbf{FNL}}_R} x \Rightarrow a$ if and only if $\Phi \models_{\mathsf{RLUG}_{R}} x \Rightarrow a$.
\end{Lem}

Next, as a typical example of a class of $\mathcal{L}_{\oc}$-algebras, we introduce modal $r \ell u$-groupoids. 
Modal $r \ell u$-groupoids are similar to modal residuated lattices in \cite{Ono93,Ono05}. 
\begin{Def}
\label{mrlu}
A \emph{modal residuated lattice-ordered unital groupoid} (\emph{modal $r\ell u$-groupoid} for short) is an $\mathcal{L}_{\oc}$-algebra ${\mathbf{A}}=(A,\wedge,\vee,\cdot,\backslash,/,\oc,1)$ such that:
\begin{itemize}
\item $(A,\wedge,\vee,\cdot,\backslash,/,1)$ is an $r\ell u$-groupoid, and
\item the following hold:
\begin{enumerate}[(iii)]
\item $1 \leq \oc 1$,
\item $x \leq y \Rightarrow \oc x \leq \oc y$,
\item $\oc x \oc y \leq \oc (xy)$.
\end{enumerate}
\end{itemize} 
\end{Def}

The following lemma guarantees that the above condition (ii) can be replaced by the equation $\oc (x \wedge y) \leq  \oc y$.
\begin{Prop}[Cf. Ono (2005)]
\label{mrlu}
Let $\mathbf{A}$ be an $\mathcal{L}_{\oc}$-algebra. 
$\mathbf{A}$ is a modal $r \ell u$-groupoid if and only if the $\mathcal{L}$-reduct of $\mathbf{A}$ is an $r\ell u$-groupoid and the following three identities hold:
\begin{enumerate}[(iii)]
\item $1 \leq \oc 1$,
\item $\oc (x \wedge y) \leq  \oc y$,
\item $\oc x \oc y \leq \oc (xy)$.
\end{enumerate}
\end{Prop}
\begin{proof}
Let $\mathbf{A}$ be a modal $r\ell u$-groupoid. 
Obviously, the equation $x \wedge y \leq y$ holds. 
By monotonicity of $\oc$, we have $\oc (x \wedge y)\leq \oc y$. 
Conversely, let $\mathbf{A}$ be an $\mathcal{L}_{\oc}$-algebra in which the equations (i), (ii) and (iii) in the statement hold. 
Suppose that $x \leq y$, i.e., $x=x \wedge y$. Using the identity (ii), we have $\oc x=\oc(x \wedge y)\leq \oc y$; thus $\oc x \leq \oc y$.
\end{proof}

By Proposition \ref{mrlu}, the class of modal $r\ell u$-groupoids forms a variety. 
Moreover, we introduce the algebraic semantics for $\mathbf{NACILL}$.
\begin{Def}
\label{NACILL}
An \emph{NACILL-algebra} is a modal $r \ell u$-groupoid satisfying the following identities: 
\begin{enumerate}[(vii)]
\item $\oc x \leq x$,
\item $\oc x \leq \oc \oc x$,
\item $\oc x \leq 1$,
\item $\oc x \leq \oc x\oc x$,
\item $\oc xy =y\oc x$,
\item $\oc x(yz)=(\oc xy)z$,
\item $x(y\oc z)=(xy)\oc z$.
\end{enumerate}
\end{Def}

Notice that the equation $x(\oc y z)=(x \oc y)z$ holds in NACILL-algebras.
The variety of NACILL-algebras is denoted by $\mathsf{NACILL}$. 
As in the case of $r \ell u$-groupoids, $\mathsf{NACILL}_{\mathsf{R}}$ denotes the subvariety of $\mathsf{NACILL}$ determined by $R \subseteq \{\exch,\cont,\weak\}$. 
A member of $\mathsf{NACILL}_{\mathsf{R}}$ is called an \emph{NACILL$_{R}$-algebra}. 
 One proves the following strong completeness theorem by a tedious completeness argument.   
\begin{Lem} 
\label{algebra2}
Let $R$ be a subset of $\{\exch,\cont,\weak\}$ and $\Phi \cup \{x \Rightarrow a\}$ a set of $\mathcal{L}_{\oc}$-sequents. $\Phi \vdash_{{\mathbf{NACILL}}_R} x \Rightarrow a$ if and only if $\Phi \models_{\mathsf{NACILL}_{R}} x \Rightarrow a$.
\end{Lem}

In the rest of this section, we recall the notions of nuclei, residuated frames, and Dedekind-MacNeille completions,  which are useful for proving cut elimination for a wide range of substructural logics; see e.g., \cite{CGT11,CGT12,GJ13} for more on algebraic cut elimination. 

\begin{Def}
Let ${\mathbf{G}}=(G,\cdot,\leq)$ be a partially-ordered groupoid. A map $\gamma \colon G \rightarrow G$ is called a \emph{nucleus} on $\mathbf{G}$ if it satisfies the following four conditions: for any $x,y \in G$,
\begin{description}
\item[\normalfont ($\gamma$1)] $x \leq \gamma(x)$,
\item[\normalfont ($\gamma$2)] $\gamma(\gamma(x)) \leq \gamma(x)$,
\item[\normalfont ($\gamma$3)] $x \leq y \Rightarrow \gamma(x) \leq \gamma(y)$,
\item[\normalfont ($\gamma$4)] $\gamma(x)\gamma(y)\leq \gamma(xy)$.
\end{description}
\end{Def}

Given an $r\ell u$-groupoid $\mathbf{G}$ and a nucleus $\gamma$ on $\mathbf{G}$, the algebra $\gamma({\mathbf{G}})=(\gamma(G),\wedge,\vee_{\gamma},\cdot_{\gamma},\backslash,/,\gamma(1))$, where $x \vee_{\gamma} y=\gamma(x \vee y)$ and $x \cdot_{\gamma} y=\gamma(xy)$, forms an $r \ell u$-groupoid. 

\begin{Def}
A \emph{unital residuated frame} (\emph{$ru$-frame} for short) is a tuple ${\mathbf{W}}=(W,W',N,\varepsilon)$ such that:
\begin{itemize}
\item $(W,\circ,\varepsilon)$ is a unital groupoid, 
\item $W'$ is a set, and
\item $N \subseteq W \times W'$ is a nuclear relation, i.e., for any $x,y \in W$ and $z \in W'$, there exist $x \dbackslash z, z \dslash y \in W'$ such that:
\[
x \circ y \N z \Longleftrightarrow y \N x \dbackslash z \Longleftrightarrow x \N z \dslash y. 
\] 
\end{itemize} 
\end{Def}

Let ${\mathbf{W}}=(W,W',N,\varepsilon)$ be an $ru$-frame. 
For any $X,Y \in {\mathcal{P}}(W)$ and $Z \in {\mathcal{P}}(W')$, define:
\begin{align*}
X^{\rhd}&:=\{a \in W' \mid \forall x \in X, x \N a\}; \\
Z^{\lhd}&:=\{a \in W \mid \forall z \in Z, a \N z\}; \\ 
X \circ Y&:=\{x \circ y \mid x \in X, y \in Y\}; \\
X\backslash Y&:=\{z \mid X \circ \{z\} \subseteq Y\}; \\
Y/X&:=\{z \mid \{z\} \circ X \subseteq Y\}.
\end{align*}
The map $\gamma_{N}$ on ${\mathcal{P}}(W)$ given by $\gamma_{N}(X)=X^{\rhd\lhd}$ is a nucleus on the $r \ell u$-groupoid $({\mathcal{P}}(W),\cap,\cup,\circ,\backslash,/,\{1\})$; see e.g., \cite[Lemma 3.36]{GJKO07} for a proof of this fact. 
A subset $X$ of $W$ is said to be \emph{Galois-closed} if $X=\gamma_{N}(X)$. 
Specifically, note that $Z^{\lhd}$ is a Galois-closed set for all $Z \subseteq W'$.   
The $r \ell u$-groupoid ${\mathbf{W}}^+=(\gamma_{N}[{\mathcal{P}}(W)],\cap,\cup_{\gamma_{N}},\circ_{\gamma_{N}},\backslash,/,\gamma_{N}(\{\varepsilon\}))$ is called the \emph{Galois algebra} of $\mathbf{W}$. 
The lattice reduct of $\mathbf{W}^+$ forms a complete lattice; hence we have the following result.
\begin{Lem}[Galatos and Jipsen (2013)]
\label{frame}
If $\mathbf{W}$ is an $ru$-frame, then $\mathbf{W}^+$ is a complete $r\ell u$-groupoid.
\end{Lem}

In particular for every $r \ell u$-groupoid $\mathbf{G}$, clearly ${\mathbf{W}_{\mathbf{G}}}=(G,G,\leq,1)$ is an $ru$-frame, where $x \dbackslash z=x\backslash z$ and $z \dslash x=z/x$. 
Moreover, define the map from $G$ to $\gamma_{\leq}[{\mathcal{P}}(G)]$ by $x \mapsto \{x\}^{\lhd}$. 
It is an embedding of $\mathbf{G}$ into ${\mathbf{W}^+_{\mathbf{G}}}=(\gamma_{\leq}[{\mathcal{P}}(G)],\cap,\cup_{\gamma_{\leq}},\circ_{\gamma_{\leq}},\backslash,/,\gamma_{\leq}(\{1\}))$, which preserves existing meets and joins. 
(For a proof, see e.g., Section 3.4.12 in Galatos \emph{et al.} (2007).) 
In this case, the Galois algebra $\mathbf{W}^+_{\mathbf{G}}$ is called the \emph{Dedekind-MacNeille completion} of $\mathbf{G}$. 
We say that a class $\mathcal{K}$ of $r \ell u$-groupoids \emph{admits Dedekind-MacNeille completions} if $\mathbf{W}^+_{\mathbf{A}} \in \mathcal{K}$ for any $\mathbf{A} \in \mathcal{K}$.
One can easily check that the following holds: 
\begin{Lem}[Galatos and Ono (2010)]
\label{DM}
Let $R$ be a subset of $\{\exch,\cont,\weak\}$. 
$\mathsf{RLUG}_{\mathsf{R}}$ admits Dedekind-MacNeille completions.
\end{Lem}

\section{Undecidability of propositional non-associative intuitionistic linear logics}
\label{mainsection}
In this section, we prove the main theorem. 
We start with the following lemma:
\if0
Firstly, we simply show a relationship between $r \ell u$-groupoids and NACILL-algebras, which might be seemingly obvious from the syntactic point of view, but be of importance.
\begin{Lem}
\label{DM2}
Every $r \ell u$-groupoid is embedded into the $\mathcal{L}$-reduct of an NACILL-algebra. 
That is, every $r \ell u$-groupoid is the $\mathcal{L}$-subreduct of an NACILL-algebra.
\end{Lem}
\begin{proof}
Let $\mathbf{G}$ be an $r\ell u$-groupoid. 
We have the Dedekind-MacNeille completion $\mathbf{W}^+_{\mathbf{G}}$. 
As already remarked, we have the embedding of $\mathbf{G}$ into $\mathbf{W}^+_{\mathbf{G}}$. 
Define the unary operation $\oc_{\gamma_{\leq}}$ on $\gamma_{\leq}[{\mathcal{P}}(G)]$ by:
\[
\oc_{\gamma_{\leq}}X:=\gamma_{\leq}(X \cap K).
\]
Here, $K=\{x \in G \mid x\,\, \text{satisfies the conditions (i)--(v)}\}$:
\begin{enumerate}[(iii)]
\item $x \leq 1$,
\item $xx=x$,
\item $xa=ax$, for every $a \in G$,
\item $x(ab)=(xa)b$, for every $a,b \in G$,
\item $(ab)x=a(bx)$, for every $a,b \in G$.
\end{enumerate}
Note that $K$ is closed under the multiplication of $\mathbf{G}$ and actually forms a subalgebra of the $\{\cdot,1\}$-reduct of $\mathbf{G}$. 
It can be easily checked that the $\mathcal{L}_{\oc}$-algebra ${\mathbf{W}^+_{\mathbf{G}}}=(\gamma_{\leq}[{\mathcal{P}}(G)],\cap,\cup_{\gamma_{\leq}}\cdot_{\gamma_{\leq}},\backslash,/,\oc_{\gamma_{\leq}},\gamma_{\leq}(\{1\}))$ is a complete NACILL-algebra. 
\end{proof} 
\fi

\begin{Lem}
\label{sc}
Let $R$ be a subset of $\{\exch,\cont, \weak\}$ and $\Phi \cup \{x \Rightarrow a\}$ a set of $\mathcal{L}$-sequents. $\Phi \vdash_{{\mathbf{FNL}}_R} x \Rightarrow a$ if and only if $\Phi \vdash_{{\mathbf{NACILL}}_R} x \Rightarrow a$.
\end{Lem}
\begin{proof}
It suffices to show the ``if'' direction, since the ``only-if'' direction clearly holds. 
Suppose that $\Phi \not \vdash_{{\mathbf{FNL}}_R} x \Rightarrow a$. 
By Lemma \ref{algebra1}, we have $f(\rho(y))\leq f(b)$ for all $y \Rightarrow b \in \Phi$ and  $f(\rho(x)) \not\leq f(a)$, for some $\mathbf{A} \in \mathsf{RLUG}_{\mathsf{R}}$ and some valuation $f$ into $\mathbf{A}$. 
Then we have the Dedekind-MacNeille completion $\mathbf{W}^+_{\mathbf{A}}$ of $\mathbf{A}$. 
As we have already remarked, the map $h \colon A \rightarrow \gamma_{\leq}[{\mathcal{P}}(A)]$ defined by $h(x)=\{x\}^{\lhd}$ is an embedding of $\mathbf{A}$ into $\mathbf{W}^+_{\mathbf{A}}$.    
Define the unary operation $\oc_{\gamma_{\leq}}$ on $\gamma_{\leq}[{\mathcal{P}}(A)]$ by:
\[
\oc_{\gamma_{\leq}}X:=\gamma_{\leq}(X \cap K).
\]
Here, $K=\{x \in A \mid x\,\, \text{satisfies the conditions (1)--(5)}\}$:
\begin{enumerate}[(iii)]
\renewcommand{\labelenumi}{(\arabic{enumi})}
\item $x \leq 1$,
\item $xx=x$,
\item $xa=ax$, for every $a \in A$,
\item $x(ab)=(xa)b$, for every $a,b \in A$,
\item $(ab)x=a(bx)$, for every $a,b \in A$.
\end{enumerate}
Notice that $K$ is a subalgebra of the $\{\cdot,1\}$-reduct of $\mathbf{A}$. 
Moreover, we check that the $\mathcal{L}_{\oc}$-algebra ${\mathbf{W}^{\oc+}_{\mathbf{A}}}=(\gamma_{\leq}[{\mathcal{P}}(A)],\cap,\cup_{\gamma_{\leq}},\circ_{\gamma_{\leq}},\backslash,/,\oc_{\gamma_{\leq}},\gamma_{\leq}(\{1\}))$ is an NACILL-algebra. 

Firstly, we show that $\mathbf{W}^{\oc+}_{\mathbf{A}}$ satisfies the conditions (i), (ii) and (iii) in Definition \ref{mrlu}. 
\begin{enumerate}[(iii)]
\item Clearly, $1 \in \gamma_{\leq}(\{1\}) \cap K$. By monotonicity of $\gamma_{\leq}$, $\gamma_{\leq}(\{1\}) \subseteq \oc_{\gamma_{\leq}}\gamma_{\leq}(\{1\})$.
\item Let $X,Y \in \gamma_{\leq}[{\mathcal{P}}(A)]$ be such that $X \subseteq Y$. We have $X \cap K \subseteq Y \cap K$. Clearly, $\oc_{\gamma_{\leq}} X \subseteq \oc_{\gamma_{\leq}}Y$. 
\item  Let $X,Y \in \gamma_{\leq}[{\mathcal{P}}(A)]$. 
The following inclusions hold: 
\begin{align*}
(X \cap K) \circ (Y \cap K) & \subseteq X \circ Y \\
(X \cap K) \circ (Y \cap K) & \subseteq K \circ K 
\end{align*}
$K$ is closed under multiplication; thus we have $(X \cap K) \circ (Y \cap K) \subseteq (X \circ Y) \cap K$.
Using properties of $\gamma_{\leq}$, we have $\oc_{\gamma_{\leq}} X \circ_{\gamma_{\leq}} \oc_{\gamma_{\leq}} Y \subseteq \oc_{\gamma_{\leq}}(X \circ_{\gamma{\leq}} Y)$. 
\end{enumerate}

Secondly, we show that all the equations in Definition \ref{NACILL} hold in $\mathbf{W}^{\oc+}_{\mathbf{A}}$. 
\begin{enumerate}[(vii)]
\item Let $X \in \gamma_{\leq}[{\mathcal{P}}(A)]$. 
Trivially, the inclusion $X \cap K \subseteq X$ holds. Clearly, $\oc_{\gamma_{\leq}}X \subseteq X$. 
\item  Let $X \in \gamma_{\leq}[{\mathcal{P}}(A)]$. Obviously, $X \cap K \subseteq \gamma_{\leq}(X \cap K) \cap K$. 
By monotonicity of $\gamma_{\leq}$, we have $\oc_{\gamma_{\leq}} X \subseteq \oc_{\gamma_{\leq}} \oc_{\gamma_{\leq}} X$.  
\item Let $X \in \gamma_{\leq}[{\mathcal{P}}(A)]$ and $x \in X \cap K$. 
By the definition of $K$, $x \leq 1$; hence $x \in \{1\}^{\lhd}=\gamma_{\leq}(\{1\})$. 
Thus $X \cap K \subseteq \gamma_{\leq}(\{1\})$. 
By monotonicity and idempotency of $\gamma_{\leq}$, $\oc_{\gamma_{\leq}} X \subseteq \gamma_{\leq}(\{1\})$. 
\item Let $X \in \gamma_{\leq}[{\mathcal{P}}(A)]$ and $x \in X \cap K$. 
By the definition of $K$, we have $x=xx \in (X \cap K) \circ (X \cap K)$. 
Thus $X \cap K \subseteq (X \cap K) \circ (X \cap K)$.
By using properties of nuclei, $\oc_{\gamma_{\leq}}X \subseteq \oc_{\gamma_{\leq}}X \circ_{\gamma{\leq}} \oc_{\gamma_{\leq}}X$.
\item Let $X,Y \in \gamma_{\leq}[{\mathcal{P}}(A)]$, $x \in X \cap K$ and $y \in Y$, i.e., $xy \in (X \cap K) \circ Y$. 
By the definition of $K$, $xy=yx \in Y \circ (X \cap K)$; thus $(X \cap K) \circ Y \subseteq Y \circ (X \circ K)$. 
By properties of nuclei, $\oc_{\gamma_{\leq}}X \circ_{\gamma{\leq}} Y \subseteq Y \circ_{\gamma{\leq}} \oc_{\gamma_{\leq}}X$. 
We have the converse inclusion in a similar way. 
\item Suppose that $X,Y,Z \in \gamma_{\leq}[{\mathcal{P}}(A)]$. 
Let $x \in X \cap K$, $y \in Y$ and $z \in Z$, i.e., $x(yz) \in (X \cap K) \circ (Y \circ Z)$. 
We have $x(yz)=(xy)z \in ((X \cap K) \circ Y) \circ Z$; thus $ (X \cap K) \circ (Y \circ Z) \subseteq ((X \cap K) \circ Y) \circ Z$.  
We have $\oc_{\gamma_{\leq}} X \circ_{\gamma_{\leq}}(Y \circ_{\gamma_{\leq}} Z) \subseteq (\oc_{\gamma_{\leq}} X \circ_{\gamma_{\leq}}Y) \circ_{\gamma_{\leq}} Z$, using properties of nuclei. 
Similarly, one can show the reverse inclusion. 
\item Almost the same as (vi).
\end{enumerate}

Consequently, $h$ is an embedding of $\mathbf{A}$ into the $\mathcal{L}$-reduct of $\mathbf{W}^{\oc+}_{\mathbf{A}}$. 
By Lemma \ref{DM},  the $\mathcal{L}$-reduct of $\mathbf{W}^{\oc+}_{\mathbf{A}}$ belongs to $\mathsf{RLUG}_{R}$. 
Thus $\mathbf{W}^{\oc+}_{\mathbf{A}}$ is an NACILL$_{R}$-algebra. 
Clearly, $h(f(\rho(y)))\subseteq h(f(b))$ for all $y \Rightarrow b \in \Phi$.
Due to the fact that $h$ is an embedding of $\mathbf{A}$ into $\mathbf{W}^{+}_{\mathbf{A}}$, we have $h(f(\rho(x))) \not \subseteq h(f(a))$. 
Moreover, define the valuation $v$ into $\mathbf{W}^{\oc+}_{\mathbf{A}}$ by $v(p)=h(f(p))$ for any propositional variable $p$. 
%It is naturally extended to the homomorphism $\mathbf{Fm}_{\mathcal{L}_{\oc}}$ into $\mathbf{W}^+_{\mathbf{A}}$.  
Clearly, $v(c)=h(f(c))$ for any $c \in Fm_{\mathcal{L}}$.
Hence, $v$ is a valuation into $\mathbf{W}^{\oc+}_{\mathbf{A}}$ such that $v(\rho(x)) \not \subseteq v(a)$ and $v(\rho(y)) \subseteq v(b)$ for all $y \Rightarrow b \in \Phi$, i.e., $\Phi \not \models_{{\mathbf{W}}^{\oc+}_{\mathbf{A}},v} x \Rightarrow a$.
By Lemma \ref{algebra2}, $\Phi \not\vdash_{{\mathbf{NACILL}}_R} x \Rightarrow a$. 
\end{proof}

Given a sequent $s=x \Rightarrow a$, define $\tau(s)=\oc(\rho(x) \backslash a)$.  
Next, we prove the following lemma:

\begin{Lem}
\label{encoding}
Let $R$ be a subset of $\{\exch,\cont,\weak\}$ and $\{s_1,\ldots,s_n\} \cup \{x \Rightarrow a\}$ a finite set of $\mathcal{L}_{\oc}$-sequents.
Then $\{s_1,\ldots,s_n\} \vdash_{{\mathbf{NACILL}}_R} x \Rightarrow a$ if and only if $\vdash_{{\mathbf{NACILL}}_R} x \circ (\tau(s_1) \circ \cdots(\tau(s_{n-1}) \circ \tau(s_n))\cdots) \Rightarrow a.$
\end{Lem}
\begin{proof}
The proof depends on the argument in \cite[Theorem 3]{Lin92}. 
We show the ``only-if'' direction by induction on the length of a proof $\mathsf{P}$ of $x \Rightarrow a$ in ${\mathbf{NACILL}}_{R}$ from $\{s_1,\ldots,s_n\}$. 
\begin{itemize}
\item If $\mathsf{P}$ is of the form {\AxiomC{$s_i=x_i \Rightarrow a_i$}\DisplayProof} ($i=1,\ldots,n$), we have:
\begin{prooftree}
\AxiomC{$\vdots$}
\noLine
\UnaryInfC{$x_i \Rightarrow \rho(x_i)$}
\AxiomC{}
\RightLabel{(Id)}
\UnaryInfC{$a_i \Rightarrow a_i$}
\RightLabel{($\backslash \Rightarrow$)}
\BinaryInfC{$x_i \circ (\rho(x_i) \backslash a_i) \Rightarrow a_i$}
\RightLabel{($\oc \Rightarrow$)}
\UnaryInfC{$x_i \circ \tau(s_i) \Rightarrow a_i$}
%\RightLabel{($\oc \weak$)}
\doubleLine
\UnaryInfC{$x_i \circ (\tau(s_1) \circ \cdots(\tau(s_{n-1}) \circ \tau(s_n))\cdots) \Rightarrow a_i$}
\end{prooftree}
Here, the double line means several applications of $(\oc \weak)$. 
Note that $\vdash_{{\mathbf{NACILL}}_R}x_i \Rightarrow \rho(x_i)$. 
\item If $\mathsf{P}$ is of the form {\AxiomC{\vdots}
\noLine
\UnaryInfC{$x \Rightarrow a$}
\AxiomC{\vdots}
\noLine
\UnaryInfC{$u[b]\Rightarrow c$}
\RightLabel{$(\backslash \Rightarrow)$}
\BinaryInfC{$u[x \circ (a \backslash b)] \Rightarrow c$}
\DisplayProof}, we have:
\begin{prooftree}
\AxiomC{$x \circ (\tau(s_1) \circ \cdots(\tau(s_{n-1}) \circ \tau(s_n))\cdots) \Rightarrow a$}
\AxiomC{$u[b] \circ (\tau(s_1) \circ \cdots(\tau(s_{n-1}) \circ \tau(s_n))\cdots) \Rightarrow c$}
\RightLabel{$(\backslash \Rightarrow)$}
\BinaryInfC{$u[(x \circ (\tau(s_1) \circ \cdots(\tau(s_{n-1}) \circ \tau(s_n))\cdots)) \circ (a \backslash b)] \circ (\tau(s_1) \circ \cdots(\tau(s_{n-1}) \circ \tau(s_n))\cdots) \Rightarrow c$}
\doubleLine
\UnaryInfC{$(u[x \circ (a \backslash b)] \circ (\tau(s_1) \circ \cdots(\tau(s_{n-1}) \circ \tau(s_n))\cdots)) \circ  (\tau(s_1) \circ \cdots(\tau(s_{n-1}) \circ \tau(s_n))\cdots) \Rightarrow c$}
\RightLabel{$(\oc \asso^*)$}
\UnaryInfC{$u[x \circ (a \backslash b)] \circ ((\tau(s_1) \circ \cdots(\tau(s_{n-1}) \circ \tau(s_n))\cdots) \circ  (\tau(s_1) \circ \cdots(\tau(s_{n-1}) \circ \tau(s_n))\cdots)) \Rightarrow c$}
\RightLabel{$(\oc \cont)$}
\UnaryInfC{$u[x \circ (a \backslash b)] \circ (\tau(s_1) \circ \cdots(\tau(s_{n-1}) \circ \tau(s_n))\cdots) \Rightarrow c$}
\end{prooftree}
\end{itemize}
At the double line in the above proof, we use the fact that $\vdash_{{\mathbf{NACILL}}_R} u[\oc a] \Rightarrow c$ if and only if $\vdash_{{\mathbf{NACILL}}_R} u[\varepsilon] \circ \oc a \Rightarrow c$. 
It is easy to show the other cases. 

The proof of the ``if'' direction is straightforward. 
Suppose that $\vdash_{{\mathbf{NACILL}}_R} x \circ (\tau(s_1) \circ \cdots(\tau(s_{n-1}) \circ \tau(s_n))\cdots) \Rightarrow a$.
Trivially, $\{s_1,\ldots,s_n\} \vdash_{{\mathbf{NACILL}}_{R}} x \circ (\tau(s_1) \circ \cdots(\tau(s_{n-1}) \circ \tau(s_n))\cdots) \Rightarrow a$. 
On the other hand, we have $\{s_1,\ldots,s_n\} \vdash_{{\mathbf{NACILL}}_{R}} \varepsilon \Rightarrow \tau(s_i)$ for $i  \in \{1,\ldots,n\}$, using the rules of ($\cdot \Rightarrow$), ($\Rightarrow \backslash$) and ($\Rightarrow \oc$). 
Hence, we have $\{s_1,\ldots,s_n\} \vdash_{{\mathbf{NACILL}}_R} x \Rightarrow a$ by applying the rule of (cut) several times. 
\end{proof}

Lemmas \ref{sc} and \ref{encoding} establish the following:

\begin{Cor}
\label{core}
Let $R$ be a subset of $\{\exch,\cont,\weak\}$ and $\{s_1,\ldots,s_n\} \cup \{x \Rightarrow a\}$ a finite set of $\mathcal{L}$-sequents. Then, $\{s_1,\ldots,s_n\} \vdash_{{\mathbf{FNL}}_R} x \Rightarrow a $ if and only if $\vdash_{{\mathbf{NACILL}}_R} x \circ (\tau(s_1) \circ \cdots(\tau(s_{n-1}) \circ \tau(s_n))\cdots) \Rightarrow a$.
\end{Cor}

On the other hand, the following theorem is shown in \cite{Ch15}:
\begin{Thm}[Chvalovsk\'{y} (2015)]
\label{ch}
Let $R$ be a subset of $\{\exch,\cont\}$. Given a finite set of $\mathcal{L}$-sequents $\Phi \cup \{s\}$, it is undecidable whether $\Phi \vdash_{{\mathbf{FNL}}_R} s$. 
\end{Thm}

By Corollary \ref{core} and Theorem \ref{ch}, we have the following main theorem:
\begin{Thm}
\label{main1}
Let $R$ be a subset of $\{\exch,\cont\}$. Given an ${\mathcal{L}}_{\oc}$-sequent $s$, it is undecidable whether $s$ is provable in ${\mathbf{NACILL}}_R$.
\end{Thm}

By Lemma \ref{algebra2}, we also have:
\begin{Cor}
\label{main2}
Let $R$ be a subset of $\{\exch,\cont\}$. $\mathsf{NACILL}_{\mathsf{R}}$ has an undecidable equational theory.
\end{Cor}

\begin{Remark}
The undecidability results proved in \cite{Ch15} are slightly different from Theorem~\ref{ch}. 
Precisely speaking, Chvalovsk\'{y} showed that the finitary consequence relation in each of the extensions of $\mathbf{FNL}$ without the constant $1$ by contraction and exchange is undecidable. 
On the other hand, his techniques work in a proof of Theorem~\ref{ch} without difficulty. 
In this sense, the establishment of Theorem~\ref{ch} is undoubtedly due to Chvalovsk\'{y}. 
\end{Remark}

\begin{Remark}
One obtains larger languages from the language $\mathcal{L}_{\oc}$ by adding some (possibly all) of the constants $0$, $\top$ and $\bot$. 
Likewise, the following inference rules can be added to any of the extensions of $\mathbf{NACILL}$: 
\[
\begin{bprooftree}
\AxiomC{}
\RightLabel{($0\Rightarrow$)}
\UnaryInfC{$0 \Rightarrow \varepsilon$}
\end{bprooftree}
\begin{bprooftree}
\AxiomC{$x \Rightarrow \varepsilon$}
\RightLabel{($\Rightarrow 0$)}
\UnaryInfC{$x \Rightarrow 0$}
\end{bprooftree}
\begin{bprooftree}
\AxiomC{}
\RightLabel{($\Rightarrow \top$)}
\UnaryInfC{$x \Rightarrow \top$}
\end{bprooftree}
\begin{bprooftree}
\AxiomC{}
\RightLabel{($\bot \Rightarrow$)}
\UnaryInfC{$u[\bot] \Rightarrow c$}
\end{bprooftree}
\]
We always assume that the right-hand side of a sequent is allowed to be $\varepsilon$ when the rules of $(0 \Rightarrow)$ and $(\Rightarrow 0)$ are added to the sequent calculus in question.  
We emphasize that Theorem~\ref{main1} holds even when some (possibly all) of the above rules are added.
\end{Remark}

\section{Undecidability of propositional non-associative non-commutative classical linear logics}
\label{classical}
In this section, we introduce two classical versions of $\mathbf{NACILL}$ and prove that both of them are undecidable. 

For the purpose of this section, we start with the logics $\mathbf{FCNL}^-$ and $\mathbf{FCNL}$, which are classical  versions of $\mathbf{FNL}$. 
The language ${\mathcal{L}}^0$ of $\mathbf{FCNL}^-$ is obtained from $\mathcal{L}$ by adding the constant $0$. 
$Fm_{{\mathcal{L}}^{0}}$ denotes the set of ${\mathcal{L}}^{0}$-formulas. 
In what follows, $a\backslash 0$ (resp. $0/a$) is abbreviated by $\negr a$ (resp. $\negl a$). 
An ${\mathcal{L}}^{0}$-structure is an element of the free unital groupoid (denoted by $Fm_{{\mathcal{L}}^0}^{\circ}$) generated by $Fm_{{\mathcal{L}}^{0}}$. 
An ${\mathcal{L}}^{0}$-sequent is denoted by $x \Rightarrow a$, where $x$ is an ${\mathcal{L}}^0$-structure, and $a$ is an ${\mathcal{L}}^0$-formula or $\varepsilon$.
A sequent calculus for $\mathbf{FCNL}^-$ is obtained from the sequent calculus for $\mathbf{FNL}$ by adding the following new inference rules:
\[
\begin{bprooftree}
\AxiomC{}
\RightLabel{($0 \Rightarrow$)}
\UnaryInfC{$0 \Rightarrow \varepsilon$}
\end{bprooftree}
\begin{bprooftree}
\AxiomC{$x \Rightarrow \varepsilon$}
\RightLabel{($\Rightarrow 0$)}
\UnaryInfC{$x \Rightarrow 0$}
\end{bprooftree}
\]
\[
\begin{bprooftree}
\AxiomC{}
\RightLabel{(DN1)}
\UnaryInfC{$\negrl a\Leftrightarrow a$}
\end{bprooftree}
\begin{bprooftree}
\AxiomC{}
\RightLabel{(DN2)}
\UnaryInfC{$\neglr a\Leftrightarrow a$}
\end{bprooftree}
\begin{bprooftree}
\AxiomC{}
\RightLabel{(CON)}
\UnaryInfC{$\negr a/b \Leftrightarrow a \backslash \!\!\negl \!b$}
\end{bprooftree}
\]
Here, the rules of (DN1), (DN2) and (CON) are originally introduced in \cite{Bus19}. 
An expression of the form $a \Leftrightarrow b$ is short for the sequents $a \Rightarrow b$ and $b \Rightarrow a$. 

The language of $\mathbf{FCNL}$ is the same as that of $\mathbf{FCNL}^-$. 
A sequent calculus for $\mathbf{FCNL}$ is obtained from the sequent calculus for $\mathbf{FCNL}^-$ by adding $\negr a \Leftrightarrow \negl a$ (cyclicity) as an initial sequent. 
The consequence relations $\vdash_{\mathbf{FCNL}^-}$ and $\vdash_{\mathbf{FCNL}}$ are defined in the obvious way. 
We stress that $\mathbf{FCNL}^-$ (resp. $\mathbf{FCNL}$) is equivalent to the logic FCNL1$^-$ (resp. FCNL1), which is introduced in \cite{Bus16}.

Next, we recall the algebraic models for $\mathbf{FCNL}^-$ and $\mathbf{FCNL}$. 
An \emph{involutive $r \ell u$-groupoid} is an algebra of the form ${\mathbf{A}}=(A,\wedge,\vee,\cdot,\backslash,/,1,0)$ (i.e., an ${\mathcal{L}}^0$-algebra) such that: 
\begin{itemize}
\item $(A,\wedge,\vee,\cdot,\backslash,/,1)$ is an $r \ell u$-groupoid, 
\item $0$ is an element of $A$, and 
\item the following identities hold:
\begin{enumerate}[(iii)]
\item $\negrl x=x=\neglr x$,
\item $\negr y/x=y \backslash \negl x$. 
\end{enumerate}
\end{itemize}
Here, $\negr x$ (resp. $\negl x$) is the abbreviation of $x\backslash 0$ (resp. $0/x$). 
Involutive $r \ell u$-groupoids are term equivalent to FCNL1$^-$-algebras in \cite{Bus16}. 
The class of involutive $r \ell u$-groupoids is a variety and is denoted by $\mathsf{InRLUG}$.
An involutive $r \ell u$-groupoid is said to be \emph{cyclic} if $\negr x=\negl x$ holds. 
Hence, the class of cyclic involutive $r \ell u$-groupoids, which is denoted by $\mathsf{CyInRLUG}$, forms a variety. 

Given an involutive $r \ell u$-groupoid $\mathbf{A}$, a valuation into $\mathbf{A}$ is a map $f \colon {\mathsf{Var}} \rightarrow A$. 
Then we have the homomorphism $f \colon {\mathbf{Fm}}_{{\mathcal{L}}^0} \rightarrow \mathbf{A}$ in a standard way.    
We define the map $\sigma: Fm_{{\mathcal{L}}^0} \cup \{\varepsilon\} \rightarrow Fm_{{\mathcal{L}}^0}$ as follows:
\[
\sigma(a) =
\begin{cases}
a & \text{if $a \in Fm_{{\mathcal{L}}^0}$,}\\
0 & \text{if $a=\varepsilon$.}
\end{cases}
\]

Given a set $\Phi \cup \{x \Rightarrow a\}$ of ${\mathcal{L}}^0$-sequents and a class $\mathcal{K}$ of involutive $r \ell u$-groupoids, we write $\Phi \models_{\mathcal{K}} x \Rightarrow a$, if for any $\mathbf{A} \in \mathcal{K}$ and any valuation $f$ into $\mathbf{A}$, $f(\rho(x)) \leq f(\sigma(a))$ holds whenever $f(\rho(y)) \leq f(\sigma(b))$ holds for each $y \Rightarrow b \in \Phi$. 
One can easily show that the following holds. 

\begin{Prop} 
Let $\Phi \cup \{x \Rightarrow a\}$ be a set of ${\mathcal{L}}^{0}$-sequents. The following statements hold:
\begin{enumerate}[(ii)]
\item $\Phi \vdash_{{\mathbf{FCNL}}^-} x \Rightarrow a$ if and only if $\Phi \models_{\mathsf{InRLUG}} x \Rightarrow a$.
\item $\Phi \vdash_{{\mathbf{FCNL}}} x \Rightarrow a$ if and only if $\Phi \models_{\mathsf{CyInRLUG}} x \Rightarrow a$.
\end{enumerate}
\end{Prop} 

Moreover, we recall involutive versions of residuated frames, based on \cite{GJ13}. 
An \emph{involutive ru-frame} is an expression of the form ${\mathbf{W}}=(W,W,N,\varepsilon,^{\sim},^{-})$ such that $(W,W,N,\varepsilon)$ is an $ru$-frame, and $^{\sim}$ and $^-$ are unary operations on $W$ satisfying the following three conditions: for any $x,y \in W$,
\begin{enumerate}[(iii)]
\item $x \dbackslash y=(y^{-} \circ x)^{\sim}$ and $y \dslash x=(x \circ y^{\sim})^{-}$,
\item $x^{\sim-}=x=x^{-\sim}$,
\item $(x^{\sim} \circ y^{\sim})^{-}=(x^{-} \circ y^{-})^{\sim}$.
\end{enumerate}

For instance, given an involutive $r \ell u$-groupoid $\mathbf{G}$, ${\mathbf{W}_{\mathbf{G}}}=(G,G,\leq,1,\sim,-)$ is an involutive $ru$-frame. 
An involutive $ru$-frame $\mathbf{W}$ is said to be \emph{cyclic} if $x^{\sim}=x^{-}$ for all $x \in W$. 
Given an involutive $ru$-frame ${\mathbf{W}}=(W,W,N,\varepsilon,^{\sim},^{-})$, we define $X^{\sim}=\{x^{\sim} \mid x \in X\}$ and $X^{-}=\{x^{-} \mid x \in X\}$ for all $X \subseteq W$. 
Likewise, define $\negr X=X^{\sim\lhd}$ and $\negl X=X^{-\lhd}$. 
Then one has the Galois algebra ${\mathbf{W}}^+=(\gamma_{N}[{\mathcal{P}}(W)],\cap,\cup_{\gamma_{N}},\circ_{\gamma_{N}},\backslash,/,\gamma_{N}(\{\varepsilon\}),0_{\gamma_{N}})$, where $0_{\gamma_{N}}=\,\,\negr \{\varepsilon\}=\negl\{\varepsilon\}$. 
Observe that $0_{\gamma_{N}}=\,\,\negr \gamma_{N}(\{\varepsilon\})=\negl \gamma_{N}(\{\varepsilon\})$. 
Specifically, given an involutive $r \ell u$-groupoid $\mathbf{G}$, the Galois algebra $\mathbf{W}^+_{\mathbf{G}}$ is called the Dedekind-MacNeille completion of $\mathbf{G}$. 
The following lemma holds:
\begin{Lem}
\label{inv1}
Let $\mathbf{W}$ be an involutive $ru$-frame. The Galois algebra $\mathbf{W}^+$ is a complete involutive $r \ell u$-groupoid. 
\end{Lem}

In particular, if $\mathbf{W}$ is cyclic, then $\mathbf{W}^+$ is a complete cyclic involutive $r \ell u$-groupoid.
Moreover, one can prove the following lemmas.  

\begin{Lem}
\label{inv2}
Let $\mathbf{G}$ be an involutive $r \ell u$-groupoid. Then the map $x \mapsto \{x\}^{\lhd}$ is an embedding of $\mathbf{G}$ into $\mathbf{W}_{\mathbf{G}}^+$.
\end{Lem}
\begin{Lem}
\label{InDM}
$\mathsf{InRLUG}$ and $\mathsf{CyInRLUG}$ admit Dedekind-MacNeille completions. 
\end{Lem}

For the proofs of Lemmas \ref{inv1}, \ref{inv2} and \ref{InDM}, the reader is referred to \cite[Section~4]{GJ13}. 

Now we are ready to introduce propositional non-associative non-commutative classical linear logics $\mathbf{NACCLL}^-$ and $\mathbf{NACCLL}$. 
The language ${\mathcal{L}}^{0}_{\oc}$ of these two logics is obtained from ${\mathcal{L}}^{0}$ by adding the unary operation symbol $\oc$. 
Formulas, structures and sequents in the language ${\mathcal{L}}^{0}_{\oc}$ (i.e., ${\mathcal{L}}^{0}_{\oc}$-formulas, ${\mathcal{L}}^{0}_{\oc}$-structures and ${\mathcal{L}}^{0}_{\oc}$-sequents) are defined in the same way as those in ${\mathcal{L}}^{0}$. 
Note that the succedent of an ${\mathcal{L}}^0_{\oc}$-sequent is allowed to be $\varepsilon$. 
A sequent calculus for $\mathbf{NACCLL}^-$ is obtained from the sequent calculus for $\mathbf{FCNL}^-$ by adding all the inference rules in Figure~\ref{inf2}; see Section~\ref{pre}. 
A sequent calculus for $\mathbf{NACCLL}$ can be also obtained from the sequent calculus for $\mathbf{NACCLL}^-$ by adding the rule of cyclicity.

Moreover, we introduce the algebraic semantics for $\mathbf{NACCLL}^-$ and $\mathbf{NACCLL}$. 
An \emph{NACCLL$^-$-algebra} is an algebra of the form $(A,\wedge,\vee,\cdot,\backslash,/,\oc,1,0)$ such that $(A,\wedge,\vee,\cdot,\backslash,/,\oc,1)$ is an NACILL-algebra and $(A,\wedge,\vee,\cdot,\backslash,/,1,0)$ is an involutive $r \ell u$-groupoid.
An \emph{NACCLL-algebra} is just a cyclic NACCLL$^-$-algebra. 
$\mathsf{NACCLL}^-$ (resp. $\mathsf{NACCLL}$) denotes the variety of NACCLL$^-$-algebras (resp. NACCLL-algebras).
Then we have:
\begin{Lem} 
\label{algebra3}
Let $\Phi \cup \{x \Rightarrow a\}$ be a set of ${\mathcal{L}}_{\oc}^{0}$-sequents. The following statements hold:
\begin{enumerate}[(ii)]
\item $\Phi \vdash_{{\mathbf{NACCLL}}^-} x \Rightarrow a$ if and only if $\Phi \models_{\mathsf{NACCLL}^-} x \Rightarrow a$.
\item $\Phi \vdash_{{\mathbf{NACCLL}}} x \Rightarrow a$ if and only if $\Phi \models_{\mathsf{NACCLL}} x \Rightarrow a$.
\end{enumerate}
\end{Lem}

The following theorem plays a critical role in proofs of the undecidability of $\mathbf{NACCLL}^-$ and $\mathbf{NACCLL}$. 
\begin{Thm}[Buszkowski (2016)]
\label{Bus}
Let $\Phi \cup \{s\}$ be a finite set of ${\mathcal{L}}^{0}$-sequents. It is undecidable whether $\Phi \vdash_{\mathbf{FCNL}^-} s$. Also, it is undecidable whether $\Phi \vdash_{\mathbf{FCNL}} s$.
\end{Thm} 

Buszkowski proved this, using the fact that both of the logics $\mathbf{FCNL}^-$ and $\mathbf{FCNL}$ are  strongly conservative extensions of $\mathbf{FNL}$; see \cite{Bus16}. 
In view of Theorem~\ref{Bus}, we prove the undecidability of $\mathbf{NACCLL}^-$ and $\mathbf{NACCLL}$ by showing the following two lemmas. 

\begin{Lem}
\label{classical1}
Let $\Phi \cup \{x \Rightarrow a\}$ be a set of ${\mathcal{L}}^{0}$-sequents. The following statements hold:
\begin{enumerate}[(ii)]
\item $\Phi \vdash_{{\mathbf{FCNL}}^-} x \Rightarrow a$ if and only if $\Phi \vdash_{{\mathbf{NACCLL}}^-} x \Rightarrow a$.
\item $\Phi \vdash_{{\mathbf{FCNL}}} x \Rightarrow a$ if and only if $\Phi \vdash_{{\mathbf{NACCLL}}} x \Rightarrow a$.
\end{enumerate}
\end{Lem}

Given an ${\mathcal{L}}^0_{\oc}$-sequent $s=x \Rightarrow a$, define $\tau^*(s)=\oc(\rho(x)\backslash \sigma(a))$. 
\begin{Lem}
\label{classical2}
Let $\{s_1,\ldots,s_n\} \cup \{x \Rightarrow a\}$ be a finite set of ${\mathcal{L}}^{0}_{\oc}$-sequents. The following statements hold:
\begin{enumerate}[(ii)]
\item $\{s_1,\ldots,s_n\} \vdash_{{\mathbf{NACCLL}}^-} x \Rightarrow a$ if and only if $\vdash_{{\mathbf{NACCLL}}^-} x \circ (\tau^*(s_1) \circ \cdots(\tau^*(s_{n-1}) \circ \tau^*(s_n))\cdots) \Rightarrow a$.
\item $\{s_1,\ldots,s_n\} \vdash_{{\mathbf{NACCLL}}} x \Rightarrow a$ if and only if $\vdash_{{\mathbf{NACCLL}}} x \circ (\tau^*(s_1) \circ \cdots(\tau^*(s_{n-1}) \circ \tau^*(s_n))\cdots) \Rightarrow a$.
\end{enumerate}
\end{Lem}

The proofs of Lemmas \ref{classical1} and \ref{classical2} are essentially the same as those of Lemmas \ref{sc} and \ref{encoding} in Section \ref{mainsection}. 
%Specifically, Lemma \ref{classical1} is shown by using Lemmas \ref{inv2}, \ref{InDM} and \ref{algebra3}.    
Consequently, we have:
\begin{Thm}
The decision problems for $\mathbf{NACCLL}^-$ and $\mathbf{NACCLL}$ are undecidable. 
\end{Thm}

\begin{Remark}
Unfortunately, none of our sequent calculi for $\mathbf{FCNL}^-$, $\mathbf{FCNL}$, $\mathbf{NACCLL}^-$ and $\mathbf{NACCLL}$ admit cut elimination. 
For our purpose here, however, we consider that our formulation of these logics is more convenient than other formulations, such as dual Sch\"{u}tte style systems in \cite{Bus16}, because our formulation allows us to obtain classical non-associative logics from $\mathbf{FNL}$ and $\mathbf{NACILL}$ by merely adding several inference rules. 
\end{Remark}

\section{Concluding remarks and future work}
\label{future}
We have introduced a non-associative and non-commutative version of propositional intuitionistic linear logic and have shown that all its extensions by exchange and contraction are undecidable. 
Likewise, we have also shown the undecidability of propositional non-associative non-commutative classical linear logics.

In Section~\ref{mainsection}, we have employed algebraic techniques to prove the main result, since it seemed difficult to prove Lemma~\ref{sc}, using only proof-theoretic methods. 
We believe that our undecidability results can be shown by purely proof-theoretic methods. 

In connection with the problems we dealt with in this paper, the following two questions remain open:
\begin{enumerate}[(ii)]
\item Are the logics $\mathbf{NACILL}_{\weak}$ and $\mathbf{NACILL}_{\exch\!\weak}$ decidable?
\item Are the logics $\mathbf{FNL}_{\cont}$ and $\mathbf{FNL}_{\exch\!\cont}$ undecidable?
\end{enumerate}

Regarding the question (i), the techniques described in Section~\ref{mainsection} cannot be used to show that $\mathbf{NACILL}_{\weak}$ and $\mathbf{NACILL}_{\exch\!\weak}$ are undecidable, because Blok and van Alten (2002) proved that the deducibility problems for the multiplicative-additive fragments of both of these logics are already decidable. 
We conjecture that $\mathbf{NACILL}_{\weak}$ and $\mathbf{NACILL}_{\exch\!\weak}$ are decidable. %, as well as affine logic \cite. 

The question (ii) is also mentioned in \cite{Ch15} and seems to be of interest to substructural logicians rather than linear logicians. 
We conjecture that  $\mathbf{FNL}_{\cont}$ and $\mathbf{FNL}_{\exch\!\cont}$ are undecidable. 
If $\mathbf{FNL}_{\cont}$ and $\mathbf{FNL}_{\exch\!\cont}$ are undecidable, one has the undecidability of  $\mathbf{NACILL}_{\cont}$ and $\mathbf{NACILL}_{\exch\!\cont}$, due to the fact that $\mathbf{NACILL}_{\cont}$ (resp. $\mathbf{NACILL}_{\exch\!\cont}$) is a conservative extension of $\mathbf{FNL}_{\cont}$ (resp. $\mathbf{FNL}_{\exch\!\cont}$). 
This fact immediately follows from Theorem \ref{cutelim} in \ref{gentzen}.

\begin{appendix}
\section{Cut elimination for propositional non-associative intuitionistic linear logics}
\label{gentzen}

\setcounter{section}{0}
\renewcommand{\thesection}{A\Alph{section}}

We present a uniform proof of cut elimination for $\mathbf{NACILL}$ and all its extensions by the rules of exchange, contraction and weakening, using modal residuated frames. 
%The algebraic account for cut elimination via Gentzen frames allows us to avoid the syntactic obstacles caused from the absence of associativity and the presence of a modality. 
Our proof is a slight refinement of the algebraic proof of the cut elimination for $\mathbf{FNL}$, which was given in \cite{GJ13}.

First of all, we introduce modal $ru$-frames. 
\begin{Def}
A \emph{modal $ru$-frame} is a tuple ${\mathbf{W}}=(W,W',N,\varepsilon,K)$ such that $(W,W',N,\varepsilon)$ is an $ru$-frame and $K$ is a subalgebra of $(W,\circ,\varepsilon)$. 
\end{Def}

Given a modal $ru$-frame $\mathbf{W}$, define the unary operation $\oc_{\gamma_{N}}$ on $\gamma_{N}[{\mathcal{P}}(W)]$ by $\oc_{\gamma_{N}} X=\gamma_{N}(X \cap K)$. 
The following lemma holds: 
\begin{Lem}
\label{mru}
If $\mathbf{W}$ is a modal $ru$-frame, ${\mathbf{W}}^+=(\gamma_{N}[{\mathcal{P}}(W)],\cap,\cup_{\gamma_{N}},\circ_{\gamma_{N}},\backslash,/,\oc_{\gamma_{N}},\gamma_{N}(\{\varepsilon\}))$ is a complete modal $r \ell u$-groupoid.  
In addition, the identities $\oc x \leq x$ and $\oc x \leq \oc\oc x$ hold in $\mathbf{W}^+$.
\end{Lem}
\begin{proof}
By Lemma \ref{frame}, we know that the $\mathcal{L}$-reduct of $\mathbf{W}^+$ is a complete $r \ell u$-groupoid. 
We show that the three conditions given in Definition \ref{mrlu} are satisfied in $\mathbf{W}^+$.
\begin{enumerate}[(iii)]
\item Obviously, $\varepsilon \in \gamma_{N}(\{\varepsilon\}) \cap K$. 
We have $\gamma_{N}(\{\varepsilon\}) \subseteq \oc_{\gamma_{N}}\gamma_{N}(\{\varepsilon\})$ by monotonicity of $\gamma_{N}$. 
\item Let $X,Y \in \gamma_{N}[{\mathcal{P}}(W)]$ be such that $X \subseteq Y$. Clearly, $X \cap K \subseteq Y \cap K$. By monotonicity of $\gamma_{N}$, $\oc_{\gamma_{N}}X \subseteq \oc_{\gamma_{N}}Y$. 
\item Let $X,Y \in \gamma_{N}[{\mathcal{P}}(W)]$.
Clearly, we have $X \cap K \circ Y \cap K \subseteq X \circ Y$ and $X \cap K \circ Y \cap K \subseteq K \circ K$. 
We have $X \cap K \circ Y \cap K \subseteq (X \circ Y) \cap K$, due to the fact that $K$ is closed under multiplication.
By properties of nuclei, we have $\oc_{\gamma_{N}} X \circ_{\gamma_{N}} \oc_{\gamma_{N}} Y \subseteq \oc_{\gamma_{N}}(X \circ_{\gamma_{N}} Y)$.
\end{enumerate}
A proof of the remaining claim is left to the reader. %that $\mathbf{W}^+$ also satisfies the identities (i) and (ii) in Definition \ref{NACILL}. 
\end{proof}

Given an NACILL-algebra $\mathbf{A}$, ${\mathbf{W}}_{\mathbf{A}}=(A,A,\leq,1,A^{\oc})$ is a modal $ru$-frame, where $A^{\oc}=\{\oc x \mid x \in A\}$. 
Note that ${\mathbf{W}}^+_{\mathbf{A}}=(\gamma_{\leq}[{\mathcal{P}}(A)],\cap,\cup_{\gamma_{\leq}},\circ_{\gamma_{\leq}},\backslash,/,\oc_{\gamma_{\leq}},\gamma_{\leq}(\{1\}))$ is an NACILL-algebra whose lattice reduct is complete. 
As in the case of $r \ell u$-groupoids, we have the following:
\begin{Lem}
Let $\mathbf{A}$ be an NACILL-algebra. The map $x \mapsto \{x\}^{\lhd}$ is an embedding of $\mathbf{A}$ into the complete NACILL-algebra $\mathbf{W}^+_{\mathbf{A}}$.  
\end{Lem}
\begin{proof}
Recall that the map $x \mapsto \{x\}^{\lhd}$ is an embedding of the $\mathcal{L}$-reduct of $\mathbf{A}$ into the $\mathcal{L}$-reduct of $\mathbf{W}^+_{\mathbf{A}}$. 
We check only that this map preserves the operation $\oc$. 
Due to the fact that $\oc x \leq x$ and $\oc x \in A^{\oc}$, $\oc x \in \{x\}^{\lhd} \cap A^{\oc}$. 
Thus $\{\oc x\}^{\lhd}=\gamma_{\leq}(\{\oc x\}) \subseteq \gamma_{\leq}(\{x\}^{\lhd} \cap A^{\oc})$, i.e., $\{\oc x\}^{\lhd} \subseteq \oc_{\gamma_{\leq}}\{x\}^{\lhd}$. 
For the reverse inclusion, let $a \in \{x\}^{\lhd} \cap A^{\oc}$. 
Then, $a=\oc b$ for some $b \in A$ and $a \leq x$. 
By monotonicity and idempotency of $\oc$, we have $a \leq \oc x$. 
This means that $a \in \{\oc x\}^{\lhd}$. 
Thus we have $\{x\}^{\lhd} \cap A^{\oc} \subseteq \{\oc x\}^{\lhd}$. 
Hence, $\oc_{\gamma_{\leq}}\{x\}^{\lhd} \subseteq \{\oc x\}^{\lhd}$. 
\end{proof}

Given an NACILL-algebra $\mathbf{A}$, the complete NACILL-algebra $\mathbf{W}^+_{\mathbf{A}}$ is also called the  \emph{Dedekind-MacNeille completion} of $\mathbf{A}$. 
We say that a class $\mathcal{K}$ of NACILL-algebras \emph{admits Dedekind-MacNeille completions} if $\mathbf{W}^+_{\mathbf{A}} \in \mathcal{K}$ for any $\mathbf{A} \in \mathcal{K}$.
The following holds:
\begin{Lem}
Let $R$ be a subset of $\{\exch,\cont,\weak\}$. $\mathsf{NACILL}_{\mathsf{R}}$ admits Dedekind-MacNeille completions.
\end{Lem}

Now we introduce cut-free Gentzen frames for $\mathbf{NACILL}$.  
\begin{Def}
A \emph{cut-free Gentzen frame} for $\mathbf{NACILL}$ is a pair $({\mathbf{W}},{\mathbf{A}})$ such that:
\begin{itemize}
\item ${\mathbf{W}}=(W,W',N,\varepsilon,K)$ is a modal $ru$-frame, 
\item $\mathbf{A}$ is an $\mathcal{L}_{\oc}$-algebra, 
\item there are injections $i \colon A \rightarrow W$, $j \colon A \rightarrow W'$ and $k \colon A^{\oc} \rightarrow K$, i.e., $A$ is identified with a subset of $W$ and a subset of $W'$, and $A^{\oc}$ is identified with a subset of $K$, and
\item the nuclear relation $N$ satisfies all the rules in Figure \ref{gen}.
\begin{figure}[t]
\[
\begin{bprooftree}
\AxiomC{}
\RightLabel{[1R]}
\UnaryInfC{$\varepsilon \N 1$}
\end{bprooftree}
\begin{bprooftree}
\AxiomC{}
\RightLabel{[Id]}
\UnaryInfC{$a \N a$}
\end{bprooftree}
\begin{bprooftree}
\AxiomC{$\varepsilon \N z$}
\RightLabel{[1L]}
\UnaryInfC{$1 \N z$}
\end{bprooftree}
\]
\[
\begin{bprooftree}
\AxiomC{$a \circ b \N z$}
\RightLabel{[$\cdot$L]}
\UnaryInfC{$a \cdot b \N z$}
\end{bprooftree}
\begin{bprooftree}
\AxiomC{$x \N a$}
\AxiomC{$y \N b$}
\RightLabel{[$\cdot$R]}
\BinaryInfC{$x \circ y \N a \cdot b$}
\end{bprooftree}
\]
\[
\begin{bprooftree}
\AxiomC{$x \N a$}
\AxiomC{$b \N z$}
\RightLabel{[$\backslash$L]}
\BinaryInfC{$x \circ (a \backslash b) \N z$}
\end{bprooftree}
\begin{bprooftree}
\AxiomC{$a \circ x \N b$}
\RightLabel{[$\backslash$R]}
\UnaryInfC{$x \N a \backslash b$}
\end{bprooftree}
\]
\[
\begin{bprooftree}
\AxiomC{$x \N a$}
\AxiomC{$b \N z$}
\RightLabel{[$/$L]}
\BinaryInfC{$(b/a)\circ x \N z$}
\end{bprooftree}
\begin{bprooftree}
\AxiomC{$x \circ a \N b$}
\RightLabel{[$/$R]}
\UnaryInfC{$x \N b/a$}
\end{bprooftree}
\]
\[
\begin{bprooftree}
\AxiomC{$a \N z$}
\RightLabel{[$\wedge$L]}
\UnaryInfC{$a \wedge b \N z$}
\end{bprooftree}
\begin{bprooftree}
\AxiomC{$b \N z$}
\RightLabel{[$\wedge$L]}
\UnaryInfC{$a \wedge b \N z$}
\end{bprooftree}
\begin{bprooftree}
\AxiomC{$x \N a$}
\AxiomC{$x \N b$}
\RightLabel{[$\wedge$R]}
\BinaryInfC{$x \N a \wedge b$}
\end{bprooftree}
\]
\[
\begin{bprooftree}
\AxiomC{$a \N z$}
\AxiomC{$b \N z$}
\RightLabel{[$\vee$L]}
\BinaryInfC{$a \vee b \N z$}
\end{bprooftree}
\begin{bprooftree}
\AxiomC{$x \N a$}
\RightLabel{[$\vee$R]}
\UnaryInfC{$x \N a \vee b$}
\end{bprooftree}
\begin{bprooftree}
\AxiomC{$x \N b$}
\RightLabel{[$\vee$R]}
\UnaryInfC{$x \N a \vee b$}
\end{bprooftree}
\]
\[
\begin{bprooftree}
\AxiomC{$a \N z$}
\RightLabel{[$\oc$L]}
\UnaryInfC{$\oc a \N z$}
\end{bprooftree}
\begin{bprooftree}
\AxiomC{$k \N a$}
\RightLabel{[$\oc$R]}
\UnaryInfC{$k \N \oc a$}
\end{bprooftree}
\begin{bprooftree}
\AxiomC{$\varepsilon \N z$}
\RightLabel{[$K$-w]}
\UnaryInfC{$k \N z$}
\end{bprooftree}
\begin{bprooftree}
\AxiomC{$k \circ k \N z$}
\RightLabel{[$K$-c]}
\UnaryInfC{$k \N z$}
\end{bprooftree}
\]
\[
\begin{bprooftree}
\AxiomC{$k \circ y \N z$}
\doubleLine
\RightLabel{[$K$-e]}
\UnaryInfC{$y \circ k \N z$}
\end{bprooftree}
\begin{bprooftree}
\AxiomC{$k \circ (x \circ y) \N z$}
\doubleLine
\RightLabel{[$K$-a]}
\UnaryInfC{$(k \circ x) \circ y \N z$}
\end{bprooftree}
\begin{bprooftree}
\AxiomC{$x \circ (y \circ k) \N z$}
\doubleLine
\RightLabel{[$K$-a$^*$]}
\UnaryInfC{$(x \circ y) \circ k \N z$}
\end{bprooftree}
\]
\caption{}
\label{gen}
\end{figure}
\end{itemize}
\end{Def}

In Figure~\ref{gen}, $x,y$ range over $W$, $z$ over $W'$, $a,b$ over $A$, and $k$ over $K$.  
Each of the rules in Figure \ref{gen} means that if the expression over the horizontal line holds, then so does the expression under the horizontal line. 
For instance, the rule of [$\vee$R] says that if $x \N a$ (or $x \N b$) holds, then so does $x \N a \lor b$.  
The rules of [Id] and [1R] always hold.
In particular, each of the rules of [$K$-e], [$K$-a] and [$K$-a$^*$] also means that if the expression under the double line holds, then so does the expression over the double line. 

Moreover, we can consider the following extra rules:
\[
\begin{bprooftree}
\AxiomC{$x \circ y \N z$}
\RightLabel{[e]}
\UnaryInfC{$y \circ x \N z$}
\end{bprooftree}
\begin{bprooftree}
\AxiomC{$x \circ x \N z$}
\RightLabel{[c]}
\UnaryInfC{$x \N z$}
\end{bprooftree}
\begin{bprooftree}
\AxiomC{$\varepsilon \N z$}
\RightLabel{[w]}
\UnaryInfC{$x \N z$}
\end{bprooftree}
\]
Given a subset $R$ of $\{\exch,\cont,\weak\}$, a cut-free Gentzen frame $(\mathbf{W},\mathbf{A})$ for ${\mathbf{NACILL}}$ is called a \emph{cut-free Gentzen frame for ${\mathbf{NACILL}}_R$} if the nuclear relation $N$ satisfies the rule of [$r$] for each $r \in R$.

\begin{Exm}
\label{exm2}
\begin{enumerate}[(ii)]
\item Given an NACILL$_{R}$-algebra $\mathbf{A}$, $(\mathbf{W}_{\mathbf{A}},\mathbf{A})$ is a cut-free Gentzen frame for ${\mathbf{NACILL}}_{R}$. 
\item Given a subset $R$ of $\{\exch,\cont,\weak\}$, consider the pair $({\mathbf{W}}^{cf}_{{\mathbf{N
ACILL}}_{R}},{\mathbf{Fm}}_{\mathcal{L}_{\oc}})$ such that:
\begin{itemize}
\item $\mathbf{Fm}_{\mathcal{L}_{{\oc}}}$ is the absolutely free algebra in the language $\mathcal{L}_{\oc}$ over $\mathsf{Var}$.
\item The tuple ${\mathbf{W}}^{cf}_{{\mathbf{NACILL}}_{R}}=(Fm^{\circ}_{\mathcal{L}_{\oc}},S_{Fm^{\circ}_{\mathcal{L}_{\oc}}} \times Fm_{\mathcal{L}_{\oc}},N^*,\varepsilon,K_{\mathcal{L}_{\oc}})$ is defined as follows. 
\begin{itemize}
\item $Fm^{\circ}_{\mathcal{L}_{\oc}}=(Fm^{\circ}_{\mathcal{L}_{\oc}},\circ,\varepsilon)$ is the free unital groupoid generated by $Fm_{\mathcal{L}_{\oc}}$.
\item $S_{Fm^{\circ}_{\mathcal{L}_{\oc}}}$ is the set of unary linear polynomials over $Fm^{\circ}_{\mathcal{L}_{\oc}}$.
\item For any $a \in Fm_{\mathcal{L}_{\oc}}$, $a$ is identified with $(\id,a)$, where $\id$ is the unary polynomial such that $\id[x]=x$.
\item $N^* \subseteq Fm^{\circ}_{\mathcal{L}_{\oc}} \times (S_{Fm^{\circ}_{\mathcal{L}_{\oc}}} \times Fm_{\mathcal{L}_{\oc}})$ is defined by:
\[
x \NN (u,a) \Longleftrightarrow  u[x] \Rightarrow a \text{\,\,is provable in ${\mathbf{NACILL}}_R$ without using (cut)}. 
\]
\item $K_{\mathcal{L}_{\oc}}$ is the free unital groupoid generated by $\{\oc a \mid a \in Fm_{\mathcal{L}_{\oc}}\}$.
\end{itemize}
\end{itemize} 
We identify $Fm_{{\mathcal{L}}_{\oc}}$ with a subset of $Fm^{\circ}_{\mathcal{L}_{\oc}}$ and $\{\oc a \mid a \in Fm_{\mathcal{L}_{\oc}}\}$ with a subset of $K_{\mathcal{L}_{\oc}}$. 
For all $x,y \in Fm_{\mathcal{L}_{\oc}}^{\circ}$ and $u \in S_{Fm_{\mathcal{L}_{\oc}}}$, define the unary linear polynomials $u_{x\circ}$ and $u_{\circ y}$ by $u_{x\circ}(y)=u(x \circ y)$ and $u_{\circ y}(x)=u(x \circ y)$. 
Obviously, the following holds:
\[
x \circ y \NN (u,a) \Longleftrightarrow y \NN (u_{x \circ}, a) \Longleftrightarrow x \NN (u_{\circ y}, a)
\]
Then $N^*$ forms a nuclear relation by setting $x \dbackslash (u,a)=(u_{x \circ},a)$ and $(u,a)\dslash y=(u_{\circ y},a)$. 
Moreover, ${\mathbf{W}}^{cf}_{{\mathbf{NACILL}}_{R}}$ is a modal $ru$-frame and $({\mathbf{W}}^{cf}_{{\mathbf{NACILL}}_R},{\mathbf{Fm}}_{\mathcal{L}_{\oc}})$ is a cut-free Gentzen frame for ${\mathbf{NACILL}}_R$. 
The verification is left to the reader.  
\end{enumerate}
\end{Exm}

Next, we prove:
\begin{Lem}
\label{preserve}
Let $R$ be a subset of $\{\exch,\cont,\weak\}$. If $({\mathbf{W}},{\mathbf{A}})$ is a cut-free Gentzen frame for ${\mathbf{NACILL}}_R$, then $\mathbf{W}^+$ is a complete NACILL$_{R}$-algebra. 
\end{Lem}
\begin{proof}
By Lemma \ref{mru}, $\mathbf{W}^+$ is a complete modal $r\ell u$-groupoid. 
To show that the $\mathcal{L}_{\oc}$-algebra ${\mathbf{W}}^+=(\gamma_{N}[{\mathcal{P}}(W)],\cap,\cup_{\gamma_{N}},\circ_{\gamma_{N}},\backslash,/,\oc_{\gamma_{N}},\gamma_{N}(\{\varepsilon\}))$ is an NACILL$_{R}$-algebra, we check that all the equations in Definition \ref{NACILL} and in the set $\mathsf{R}$ hold in $\mathbf{W}^+$.
Let us consider only some cases here. 

For the equation (vi) in Definition \ref{NACILL}, let $X,Y,Z \in \gamma_{N}[{\mathcal{P}}(W)]$ and $w \in (X \cap K) \circ (Y \circ Z)$. 
This means that $w=k \circ (y \circ z) \in  (X \cap K) \circ (Y \circ Z)$, for some $k \in X \cap K$, $y \in Y$ and $z \in Z$.
Obviously, we have $(k \circ y) \circ z \in (\oc_{\gamma_{N}}X \circ_{\gamma_{N}} Y) \circ_{\gamma_{N}} Z$. 
Let $w' \in  ((\oc_{\gamma_{N}}X \circ_{\gamma_{N}} Y) \circ Z)^{\rhd}$.
Then we have $(k \circ y) \circ z \N w'$.
Using the rule of [$K$-a], we have $k \circ (y \circ z) \N w'$, i.e., $w=k \circ (y \circ z) \in (\oc_{\gamma_{N}}X \circ_{\gamma_{N}} Y) \circ_{\gamma_{N}} Z$; hence $(X \cap K) \circ (Y \circ Z) \subseteq (\oc_{\gamma_{N}}X \circ_{\gamma_{N}} Y) \circ_{\gamma_{N}}Z$. 
Using properties of nuclei, we have $\oc_{\gamma_{N}}X \circ_{\gamma_{N}} (Y \circ_{\gamma_{N}}Z) \subseteq (\oc_{\gamma_{N}}X \circ_{\gamma_{N}} Y) \circ_{\gamma_{N}}Z$.
One has the converse inclusion in a similar way. 

For the equation (iii) in Definition \ref{NACILL}, let $k \in K$. 
Suppose that $\gamma_{N}(\{\varepsilon\}) \subseteq \{z\}^{\lhd}$; hence $\varepsilon \N z$. 
By the rule of [$K$-w], we have $k \N z$, i.e., $k \in \{z\}^{\lhd}$. 
Thus we have $k \in \bigcap \{\{z\}^{\lhd} \mid \gamma_{N}(\{\varepsilon\}) \subseteq \{z\}^{\lhd}\}$. 
Then, 
\begin{align*}
k \in \bigcap \{\{z\}^{\lhd} \mid \gamma_{N}(\{\varepsilon\}) \subseteq \{z\}^{\lhd}\} & \Longleftrightarrow k \in \bigcap \{\{z\}^{\lhd} \mid z \in \gamma_{N}(\{\varepsilon\})^{\rhd}\} \\
& \Longleftrightarrow k \N z, \text{for all $z \in \gamma_{N}(\{\varepsilon\})^{\rhd}$} \\
& \Longleftrightarrow k \in \gamma_{N}(\{\varepsilon\})^{\rhd\lhd}
\end{align*}
Hence, $K \subseteq \gamma_{N}(\{\varepsilon\})$. 
This implies that $\oc_{\gamma_{N}} X \subseteq \gamma_{N}(\{\varepsilon\})$. 
\end{proof}

The following lemma says that, given a cut-free Gentzen frame $({\mathbf{W}},{\mathbf{A}})$ for ${\mathbf{NACILL}}$, $\mathbf{A}$ is quasi-embeddable into the NACILL-algebra $\mathbf{W}^+$. 
\begin{Lem}
\label{truth}
Let $({\mathbf{W}},{\mathbf{A}})$ be a cut-free Gentzen frame for ${\mathbf{NACILL}}$. For every $a,b \in A$, and $X,Y \in \gamma_{N}[{\mathcal{P}}(W)]$, the following statements hold: 
\begin{enumerate}[(iii)]
\item $1 \in \gamma_{N}(\{\varepsilon\}) \subseteq \{1\}^{\lhd}$. 
\item If $a \in X \subseteq \{a\}^{\lhd}$ and $b \in Y \subseteq \{b\}^{\lhd}$ then $a \bullet b \in X \bullet_{{\mathbf{W}}^+} Y \subseteq \{a \bullet b \}^{\lhd}$, where $\bullet \in \{\wedge,\vee,\cdot,\backslash,/\}$ and $\bullet_{\mathbf{W}^+}$ denotes the operation on $\mathbf{W}^+$ corresponding to $\bullet$. 
\item If $a \in X \subseteq \{a\}^{\lhd}$ then $\oc a \in \oc_{\gamma_{N}} X \subseteq \{\oc a\}^{\lhd}$.
\end{enumerate}
\end{Lem}
\begin{proof}
It suffices to show the statement (iii), since the other statements are shown in \cite[Theorem 2.6]{GJ13}. 
Let $z \in X^{\rhd}$.
By assumption, we have $a \N z$. 
Using the rule of [$\oc$L], we have $\oc a \N z$, i.e., $\oc a \in X^{\rhd\lhd}=X$.
By the definition of $K$, $\oc a \in K$. 
Thus $\oc a \in X \cap K \subseteq \oc_{\gamma_{N}} X$.

Let $k \in X \cap K$.
Due to the fact that $k \in X \cap K \subseteq X \subseteq \{a\}^{\lhd}$, we have $k \N a$.
Using the rule of [$\oc$R], we have $k \N \oc a$; thus $k \in \{\oc a\}^{\lhd}$. 
Hence, we have $X \cap K \subseteq \{\oc a\}^{\lhd}$. 
By monotonicity and idempotency of $\gamma_{N}$, $\oc_{\gamma_{N}} X = \gamma_{N}(X \cap K) \subseteq \{\oc a\}^{\lhd}$.
\end{proof}

Given a cut-free Gentzen frame $({\mathbf{W}},{\mathbf{A}})$ for ${\mathbf{NACILL}}$ and a valuation $f $ into $\mathbf{A}$, define the valuation $f^*$ into $\mathbf{W}^+$ by $f^*(p)=\{f(p)\}^{\lhd}$.
Moreover, we have:
\begin{Lem}
\label{ind}
Let $({\mathbf{W}},{\mathbf{A}})$ be a cut-free Gentzen frame for ${\mathbf{NACILL}}$. Then $f(a) \in f^*(a) \subseteq \{f(a)\}^{\lhd}$ for any valuation $f$ into $\mathbf{A}$ and any $a \in Fm_{\mathcal{L}_{\oc}}$.
\end{Lem}
\begin{proof}
By induction on the length of $a$.
We show the case where $a=\oc b$. 
By the induction hypothesis, $f(b) \in f^*(b) \subseteq \{f(b)\}^{\lhd}$. 
By Lemma \ref{truth}, $\oc f(b) \in \oc_{\gamma_{N}} f^*(b) \subseteq \{\oc f(b)\}^{\lhd}$, i.e., $f(\oc b) \in f^*(\oc b) \subseteq \{f(\oc b)\}^{\lhd}$.
See \cite[Lemma 3.1]{GJ13} for the remaining cases. 
\end{proof}

We define the validity of an $\mathcal{L}_{\oc}$-sequent in a cut-free Gentzen frame for $\mathbf{NACILL}$, based on \cite{GJ13}. 
Given a cut-free Gentzen frame $(\mathbf{W},\mathbf{A})$ for $\mathbf{NACILL}$ and a map $f \colon Fm_{\mathcal{L}_{\oc}} \rightarrow A$, we inductively define the map $f^{\circ} \colon Fm^{\circ}_{\mathcal{L}_{\oc}} \rightarrow W$ by $f^{\circ}(x \circ y)=f^{\circ}(x) \circ f^{\circ}(y)$. 
Similarly, given an $\mathcal{L}_{\oc}$-algebra $\mathbf{A}$ and a valuation $f$ into $\mathbf{A}$,  the homomorphism $f^{\circ} \colon {\mathbf{Fm}}^{\circ}_{\mathcal{L}_{\oc}} \rightarrow \mathbf{A}$ is obtained by extending $f$, where $\mathbf{Fm}^{\circ}_{\mathcal{L}_{\oc}}$ denotes the absolutely free algebra in the language $\{\circ,\varepsilon\}$ over $Fm_{\mathcal{L}_{\oc}}$.  
We say that an $\mathcal{L}_{\oc}$-sequent $x \Rightarrow a$ is \emph{valid} in a cut-free Gentzen frame $({\mathbf{W}},{\mathbf{A}})$ for ${\mathbf{NACILL}}$ and write $({\mathbf{W}},{\mathbf{A}}) \models x \Rightarrow a$ if $f^{\circ}(x) \N f(a)$ holds for every valuation $f$ into $\mathbf{A}$. 
%Clearly, given an NACILL-algebra $\mathbf{A}$ and an $\mathcal{L}_{\oc}$-sequent $x \Rightarrow a$, $\models_{\mathbf{A}} x \Rightarrow a$ if and only if the cut-free Gentzen frame for ${\mathbf{NACILL}}$ $({\mathbf{W}}_{\mathbf{A}},{\mathbf{A}}) \models x \Rightarrow a$. 

Using Lemma \ref{ind}, we can show the following lemma:
\begin{Lem}
\label{cut}
Let $({\mathbf{W}},{\mathbf{A}})$ be a cut-free Gentzen frame for $\mathbf{NACILL}$ and $x\Rightarrow a$ an $\mathcal{L}_{\oc}$-sequent. If $\models_{{\mathbf{W}}^+} x \Rightarrow a$, then $({\mathbf{W}},{\mathbf{A}}) \models x \Rightarrow a$.
\end{Lem}
\begin{proof}
Essentially the same as the proof of \cite[Theorem 3.2]{GJ13}.
\end{proof}

Now we show the following:
\begin{Thm}
\label{cutelim}
Let $R$ be a subset of $\{\exch,\cont,\weak\}$. If $x \Rightarrow a$ is provable in ${\mathbf{NACILL}}_R$, then it is provable in ${\mathbf{NACILL}}_R$ without using the rule of (cut).
\end{Thm}
\begin{proof}
Suppose that $\vdash_{{\mathbf{NACILL}}_R} x \Rightarrow a$. 
By Lemma \ref{algebra2}, we have $\models_{\mathsf{NACILL}_{\mathsf{R}}} x \Rightarrow a$. 
As we have remarked in Example~\ref{exm2}, the pair $({\mathbf{W}}^{cf}_{{\mathbf{NACILL}}_{R}},{\mathbf{Fm}}_{\mathcal{L}_{\oc}})$ is a cut-free Gentzen frame for ${\mathbf{NACILL}}_{R}$. 
By Lemma \ref{preserve}, we have $\models_{{\mathbf{W}}^{cf+}_{{\mathbf{NACILL}}_R}} x \Rightarrow a$. 
By Lemma \ref{cut}, $({\mathbf{W}}^{cf}_{{\mathbf{NACILL}}_R},{\mathbf{Fm}}_{\mathcal{L}_{\oc}}) \models x \Rightarrow a$; hence $x \Rightarrow a$ is provable in ${\mathbf{NACILL}}_R$ without using (cut).
\end{proof}

\end{appendix}

\label{lastpage}
\end{document}